\documentclass[11pt]{article}
\usepackage{amsmath,amssymb,amsxtra}

\usepackage{graphicx}
\usepackage{xcolor}
\usepackage{tikz}

\usepackage[hmargin=0.12\paperwidth,vmargin=0.12\paperwidth,bindingoffset=0cm]{geometry}
\usepackage[square,sort,comma,numbers]{natbib}
\usepackage{amsthm}

\theoremstyle{plain}
  \newtheorem{thm}{Theorem}
  
  \newtheorem{lemma}[thm]{Lemma}
  \newtheorem{prop}[thm]{Proposition}
\theoremstyle{definition}
  \newtheorem{remark}[thm]{Remark}

\numberwithin{equation}{section}

\DeclareMathOperator{\tr}{Tr} 
\DeclareMathOperator*{\diag}{diag} 
\DeclareMathOperator{\sgn}{sgn} 
\DeclareMathOperator{\ePDF}{ePDF} 
\DeclareMathOperator*{\Span}{span} 

\newcommand{\MeijerG}[8][\bigg]{G^{{ #2 },{ #3 }}_{{ #4 },{ #5 }} #1( \begin{matrix} #6 \\ #7 \end{matrix}\, #1\vert\, #8 #1)}
\newcommand{\FoxH}[8][\bigg]{H^{{ #2 },{ #3 }}_{{ #4 },{ #5 }} #1( \begin{matrix} #6 \\ #7 \end{matrix}\, #1\vert\, #8 #1)}

\usepackage{hyperref}   
\hypersetup{
colorlinks=true,
linkcolor=blue,
citecolor=blue,
}

\begin{document}

\begin{center}
{\bfseries\Large Matrix product ensembles of Hermite type\\[0.7ex] and the hyperbolic Harish-Chandra--Itzykson--Zuber integral
}\\[2\baselineskip]
{\large P. J. Forrester\footnote{pjforr@unimelb.edu.au}, %
J. R. Ipsen\footnote{jesper.ipsen@unimelb.edu.au}}\\[.5\baselineskip]
{\itshape ARC Centre of Excellence for Mathematical and Statistical Frontiers,\\
School of Mathematics and Statistics, The University of Melbourne, Victoria 3010, Australia.}\\[\baselineskip]
{\large Dang-Zheng Liu\footnote{dzliu@ustc.educ.cn}}\\[.5\baselineskip]
{\itshape Key Laboratory of Wu Wen-Tsun Mathematics, CAS, School of Mathematical Sciences, University of Science and Technology of China, Hefei 230026, P.R.~China}

\end{center}

\begin{abstract}
\noindent
We investigate spectral properties of a Hermitised random matrix product which, contrary to previous product ensembles, allows for eigenvalues on the full real line. We prove that the eigenvalues form a bi-orthogonal ensemble, which reduces asymptotically to the Hermite Muttalib--Borodin ensemble. Explicit expressions for the bi-orthogonal functions as well as the correlation kernel are provided. Scaling the latter near the origin gives a limiting kernel involving Meijer $G$-functions, and the functional form of the global density is calculated. As a part of this study, we introduce a new matrix transformation which maps the space of polynomial ensembles onto itself. This matrix transformation is closely related to the so-called hyperbolic Harish-Chandra--Itzykson--Zuber integral.
\end{abstract}

\section{Introduction}

\subsection{Statement of the problem and summary of results}

Let $H$ be a matrix from the Gaussian unitary ensemble (GUE) and let each $G_i$ $(i=1,\dots, M)$ denote a complex Ginibre matrix, i.e. a matrix with i.i.d. standard complex Gaussian entries. In this paper, we investigate the eigenvalues of the Hermitised product matrix
\begin{equation}\label{W1}
W_M = G_M^\dagger \cdots G_1^\dagger H G_1 \cdots G_M
\end{equation}
under the assumption that all matrices, $H$ and $G_i$ ($i=1,\ldots,M$), are independent. We will see that the eigenvalues form a bi-orthogonal ensemble~\cite{Bo98}. Furthermore, this ensemble is closely related (in a sense that will be specified in the next subsection) to the so-called Hermite Muttalib--Borodin ensemble~\cite{Mu95,Bo98}. The latter is defined by the joint eigenvalue probability density function (PDF)
\begin{equation}\label{MB-Hermite}
\tilde P(x_1,\ldots,x_N)=\frac1{\tilde Z_N^{(M)}}
\prod_{1\leq j<k\leq N}(x_k-x_j)(x_k^{2M+1}-x_j^{2M+1})
\prod_{\ell=1}^N|x_\ell|^\alpha e^{-x_\ell^2},
\end{equation}
where $\tilde Z_N^{(M)}$ is a normalisation constant and $\alpha$ is a non-negative constant.

It will transpire that the bi-orthogonal ensemble structure associated with the eigenvalue PDF of the product matrix~\eqref{W1} is a corollary of the following more basic result.

\begin{thm}\label{T1}
Let $G$ be an $n \times N$ ($n \le N$) standard complex Gaussian matrix and let $A$ be an $n \times n$
Hermitian matrix with eigenvalues $a_1,\ldots,a_N$.
If the eigenvalues of $A$ are pairwise distinct and ordered as
\begin{equation}\label{as}
a_1<a_2<\cdots<a_{n_0}<0<a_{n_0+1}<\cdots<a_n,
\end{equation}
then the PDF of the non-zero eigenvalues of matrix $X=G^\dagger AG$ is given by
\begin{multline}\label{as1}
P^{n_0}_{n}(\{a_j\}_{j=1}^n;\{x_j\}_{j=1}^n)=\\
\prod_{l=1}^n\frac1{|a_l|}\frac{(x_l/a_l)^{N-n}}{(N-l)!}
\prod_{1\leq j<k\leq n}\frac{x_k-x_j}{a_k-a_j}\det\big[e^{-x_i/a_j}\big]_{i,j=1}^{n_0}
\det\big[e^{-x_{i}/a_{j}}\big]_{i,j=n_0+1}^{n},
\end{multline}
where
\begin{equation}\label{as2}
x_1< \cdots < x_{n_0} < 0 < x_{n_0+1} < \cdots < x_n.
\end{equation}
In particular, we see that $X$ has $n_0$ ($n-n_0$) negative (positive) eigenvalues, i.e. the same number as $A$.
The remaining $N-n$ eigenvalues are all identically zero.
\end{thm}

We remark that in Theorem \ref{T1} the case of $n<N$ remains unanswered although  this  is certainly  of high  interest; see \cite{DF06,ACLS} and references therein for a relevant  question.

The rest of this paper is organised as follows:
In Section~\ref{sec:product}, we use Theorem~\ref{T1} to find the PDF for the eigenvalues of the product~\eqref{W1} as a bi-orthogonal ensemble. Moreover, the explicit expression for the PDF is seen to reduce asymptotically to the functional form~\eqref{MB-Hermite} specifying the Hermite Muttalib--Borodin ensemble.
Explicit expressions for the bi-orthogonal functions are derived in Section~\ref{sec:biortho}. Analogous to the theory of Hermite polynomials (see e.g.~(\ref{HL}) below), we will see that it is convenient to consider bi-orthogonal functions of even and odd degree separately.
Section~\ref{sec:int} provides reformulations of the bi-orthogonal functions and the correlation kernel as integral representations, which are more suited for asymptotic analysis. These integral representations can also be expressed in terms of Meijer $G$-functions, and we will see that they are closely related to known formulae stemming from the product ensemble of Laguerre type.
The local scaling limit at the origin is derived and seen to be related to the Meijer $G$-kernel in Section~\ref{sec:hard}. This result is also compared with the local scaling limit of the Hermite Muttalib--Borodin ensemble.
Section~\ref{sec:global} includes derivations of the global spectrum of the product~\eqref{W1} as well as the Hermite Muttalib--Borodin ensembles. Since our product ensemble reduces asymptotically to the Hermite Muttalib--Borodin ensemble, they are as expected seen to have the same global spectrum, which in turn is given in terms of the Fuss--Catalan density.
Finally, Theorem~\ref{T1} is proven in the appendix. This theorem is an important result by itself. For this reason, we provide three separate proofs each with their own merits.

\subsection{First motivation: Muttalib--Borodin ensembles}
\label{sec:motivation}

Orthogonal polynomial ensembles are point processes on (a subset of) the real line with a joint distribution given by
\begin{equation}\label{jpdf-classical}
P(dx_1,\ldots,dx_n)=\frac1{Z_n}\Delta_n(\{x\})^2\prod_{k=1}^nw(x_k)dx_k,
\end{equation}
where $Z_N$ is a normalisation constant, $w(x)$ is a non-negative weight function, and $\Delta_n(\{x\})$ denotes the Vandermonde determinant,
\begin{equation}
\Delta_{n}(\{x\})=\det_{1\leq i,j\leq n}\big[x_i^{\,j-1}\big]=\prod_{1\leq i<j\leq n}(x_j-x_i).
\end{equation}
Like the corresponding moment problem, it is often useful to distinguish between models with support on a finite, semi-infinite, and double-infinite interval. The canonical examples are the Jacobi-, Laguerre-, and Hermite-ensembles summarised in Table~\ref{table:weights}. These ensembles are named according to the corresponding classical orthogonal polynomials. In fact, for $a\neq0$ the latter would more appropriately be called the generalised Hermite ensemble.
\begin{table}[htbp]
\centering
\caption{Summary of weight functions and support for the three canonical orthogonal ensembles in random matrix theory given by the joint distribution~\eqref{jpdf-classical}.}
\vspace*{.5em}
\label{table:weights}
\begin{tabular}{l@{\qquad\qquad}l@{\qquad\qquad}l}
\hline\hline
ensemble  & \hspace*{1.5em}weight & support\\ \hline \\[-.9em]
Jacobi & $w(\lambda)=\lambda^a(1-\lambda)^b$  & $\lambda\in(0,1)$ \\[.2em]
Laguerre & $w(\lambda)=\lambda^ae^{-\lambda}$ & $\lambda\in(0,\infty)$ \\[.2em]
Hermite & $w(\lambda)=|\lambda|^ae^{-\lambda^2}$ & $\lambda\in(-\infty,\infty)$\\[.1em] \hline\hline
\end{tabular}
\end{table}

In random matrix theory these three ensembles play a fundamental role as they appear as the distribution of the eigenvalues (or singular values) for the transfer (or truncated unitary) ensemble, the complex Wishart (or chiral) ensemble, and the Gaussian unitary ensemble, respectively; see e.g.~\cite{Fo10}.

A fundamental insight, which can be traced back to Wigner~\cite{Wi57}, is that the joint distribution~\eqref{jpdf-classical} allows an interpretation as the equilibrium measure for a one-dimensional gas of pairwise repulsive point particles in a confining potential.
More precisely, consider the Gibbs measure for a classical gas of $n$ point particles which are pairwise repulsive according to a two-point potential $U(x,y)$ and confined by a common one-point potential $V(x)$, i.e.
\begin{equation}\label{gibbs}
P(dx_1,\ldots,dx_n)=\frac{1}{Z_n}e^{-\beta E(x_1,\ldots,x_n)}\prod_{k=1}^n dx_k
\end{equation}
with $\beta$ denoting the inverse temperature and $E$ the energy functional
\begin{equation} \label{H1}
E(x_1,\ldots,x_n)=\frac12\sum_{k=1}^n V(x_k)-\sum_{1\leq i<j\leq n}U(x_j,x_i).
\end{equation}
We see that at $\beta=2$ (sometimes referred to as the free fermion point) the Gibbs measure~\eqref{gibbs} is identical to~\eqref{jpdf-classical} provided we set $V(\lambda)=-\log w(\lambda)$ and $U(x_i,x_j)=\log|x_j-x_i|$. In this way, the eigenvalues of random matrices relate to the Boltzmann factor of a simple statistical mechanical system with one- and two-body interactions only.

A recent development in random matrix theory is the study of exactly solvable product ensembles; see~\cite{AI15} for a review.
As an example, let $G_i$ ($i=1,\ldots,M$) be independent complex Ginibre matrices (matrices whose entries are i.i.d. standard complex Gaussians)
and consider the Hermitian product
\begin{equation}\label{old-product}
G_{M}^\dagger\cdots G_1^\dagger G_1\cdots G_M.
\end{equation}
From the work of Akemann et al.~\cite{AKW13,AIK13}, we know that the explicit
PDF for the eigenvalues of the matrix~\eqref{old-product} is
\begin{equation}\label{2}
P^{(M)}(x_1,\ldots,x_n)=\frac1{Z_n^{(M)}}\Delta_n(\{x\})\det_{1\leq i,j\leq n}\big[g_j^{(M)}(x_i)\big],
\qquad x_i > 0 \: (i=1,\dots,n)
\end{equation}
where $Z_n^{(M)}$ is a known normalisation constant and $g_j^{(M)}(x)$ ($j=1,\ldots,n$) are given by certain Meijer $G$-functions.
Generally, such PDFs are known as polynomial ensembles~\cite{KS14}.

It seems natural to ask whether these product ensembles also have (at least approximately) an interpretation as a Gibbs measure of the form~\eqref{gibbs} and~\eqref{H1}.
However, unlike the Vandermonde determinant, the determinant in~\eqref{2} cannot be evaluated as a product (for $M \ge 2$).
This prohibits a literal interpretation of the eigenvalues of~\eqref{old-product} as a statistical mechanical system
with only one- and two-body interactions. One could fear that this meant that there was no simple physical interpretation related to~\eqref{2}.
However, if we consider (\ref{2}) with each $x_j$ large, the Meijer $G$-functions can be replaced by their asymptotic approximation~\cite{Fi72}.
After a change of variables, the joint density~\eqref{2} to leading order in the asymptotic expansion becomes~\cite{FLZ15}
\begin{equation}\label{MB-laguerre}
\tilde P^{(M)}(x_1,\ldots,x_n)=
\frac1{\tilde Z_n^{(M)}}\Delta_n(\{x\})\Delta_n(\{x^M\})\prod_{k=1}^n x_k^a\,e^{-x_k},
\qquad x_k > 0 \: (k=1,\dots,n)
\end{equation}
where $a$ is a known non-negative constant.
This does correspond to the Boltzmann factor of a statistical mechanical system with one- and two-body interactions only.

A comparison between~\eqref{2} and~\eqref{MB-laguerre} can be done a posteriori.
A connection between the two ensembles was first noted by Kuijlaars and Stivigny~\cite{KS14}, who observed that the hard edge scaling limit of~\eqref{MB-laguerre} found in~\cite{Bo98} took the same functional form as the Meijer $G$-kernel found in the product ensemble~\cite{KZ14}, albeit with a different choice of parameters. Due to recent progress, even more is known about the scaling limits of both models, and their similarities. Thus it has been established that the two ensembles also share the same global spectral distribution~\cite{Mu02,BJLNS10,BBCC11,PZ11,FW15}. Furthermore, in both cases the local correlations in the bulk and near the soft edge are given by the familiar sine and Airy process, respectively~\cite{LWZ14,Zh15}.

The ensemble~\eqref{MB-laguerre} had, in fact, appeared in earlier random matrix literature.
It was first isolated by Muttalib~\cite{Mu95}, who suggested it as a naive approximation to the transmission eigenvalues in a problem about quantum transport.
A feature of the new interaction is that bi-orthogonal polynomials (rather than orthogonal polynomials) are needed in the study of correlation functions. Such bi-orthogonal ensembles were considered in greater generality by Borodin~\cite{Bo98}, who devoted special attention to PDFs
\begin{equation}\label{MB}
P(x_1,\ldots,x_n)=\frac1{Z_n}\prod_{j=1}^n w(x_l)\prod_{1\leq j<k\leq n} \big|x_k-x_j\big|\,
\big|\sgn(x_k)|x_k|^{\theta}-\sgn(x_j)|x_j|^{\theta}\big|,
\end{equation}
with $\theta>0$ and $w(x)$ representing one of the three classical weight functions from Table~\ref{table:weights}.
Following \cite{FW15}, we will henceforth refer to these ensembles as the (Jacobi, Laguerre, Hermite) Muttalib--Borodin ensembles.
We note that the awkward dependence of signs in the last factor in~\eqref{MB} disappears when the eigenvalues are non-negative (e.g. for Laguerre- and Jacobi-ensembles) and when $\theta$ is an odd integer as in~\eqref{MB-Hermite}.

At the time of their introduction,
the Muttalib--Borodin ensembles had no obvious relation to any random matrix models defined in terms of PDFs on their entries (except for the trivial case $\theta=1$), and could merely be interpreted as a simple one-parameter generalisation of the classical ensembles.
However, we now see that the Laguerre Muttalib--Borodin ensemble has a close connection
to products of complex Gaussian random matrices~\eqref{old-product} through the approximation~\eqref{MB-laguerre}.

Knowing that the Laguerre Muttalib--Borodin ensemble appears as an asymptotic approximation to the Gaussian product~\eqref{old-product}, it seems natural to ask the reverse question: \emph{Can we find product ensembles which reduce asymptotically to the Jacobi and Hermite Muttalib--Borodin ensembles?} If this is possible, it would be reasonable to say we have completed a link between the Muttalib--Borodin ensembles with classical weights and the new family of product ensembles.

For the Jacobi Muttalib--Borodin ensemble a link to products of random matrices is provided by looking at the squared singular values of a product of truncated unitary matrices~\cite{KKS15,FW15}. In this paper, it is our aim to isolate a random matrix product structure for which the eigenvalue PDF reduces asymptotically to the functional form of the Hermite Muttalib--Borodin ensemble. This construction therefore completes the correspondence between product ensembles and the three Muttalib--Borodin ensembles with classical weights, i.e. Laguerre, Jacobi, Hermite. Furthermore, the relevant product ensemble provides by itself a new interesting class of integrable models, which unlike all previous product ensembles (see~the review \cite{AI15}) allows for negative eigenvalues.

As the product ensemble in question must allow for negative eigenvalues, it is no longer sufficient to investigate Wishart-type matrices like~\eqref{old-product} which are positive-definite by construction.
It turns out that the correct structure is the Hermitised product of a GUE matrix and $M$ complex Ginibre matrices given by~\eqref{W1}.
The case $M = 1$ of~\eqref{W1} has previously been isolated in the recent paper of Kumar~\cite{Ku15} as an example
of a matrix ensemble which permits an explicit eigenvalue PDF.

\subsection{Second motivation: hyperbolic Harish-Chandra--Itzykson--Zuber integrals}

Another reason that the Hermitised random matrix product~\eqref{W1} is of particular interest is its relation to the so-called hyperbolic Harish-Chandra--Itzykson--Zuber (HCIZ) integral. By way of introduction on this point, we note that it is by now evident that the family of exactly solvable product ensembles is intimately linked to a family exactly solvable group integrals sometimes referred to as integrals of HCIZ type. For the study of products of Ginibre matrices~\eqref{old-product} it was sufficient to know the familiar (and celebrated) HCIZ integral~\cite{HC57,IZ80}:
\begin{equation}\label{HCIZ}
 \int_{U(N)/U(1)^N}e^{-\tr AVBV^{-1}}\,(V^{-1}dV)=\pi^{N(N-1)/2}
 \frac{\det[e^{-a_ib_j}]_{i,j=1}^{N}}
 {\prod_{1\leq i<j\leq N} (a_j-a_i)(b_j-b_i)},
\end{equation}
where $(V^{-1}dV)$ denotes the Haar measure on the unitary quotient group $U(N)/U(1)^N$, while $A$ and $B$ are Hermitian $N\times N$ matrices with eigenvalues $a_1<\cdots<a_N$ and $b_1<\cdots<b_N$, respectively. However, for studies of products of spherical, truncated unitary, or coupled random matrices generalisations of the HCIZ integral are needed, see~\cite{KKS15,AS16,Liu17} for the two latter cases. We emphasise that the product of truncated unitary matrices considered by Kieburg et al.~\cite{KKS15} required a previously unknown generalisation of the HCIZ integral. Likewise, our study of the Hermitised random matrix product~\eqref{W1} requires knowledge about the so-called hyperbolic HCIZ integral in which the integration on the left-hand side of~\eqref{HCIZ} should be replaced with an integration over the pseudo-unitary group (see Section~\ref{sec:fyodorov} for details). The study of such hyperbolic group integrals was initiated by Fyodorov~\cite{Fy02,FS02}.
An interesting feature of the hyperbolic HCIZ integral is that the integration over the pseudo-unitary is non-compact, which forces us to introduce some additional constraints on the Hermitian matrices $A$ and $B$ to ensure convergence; this is a difficulty which does not arise in other HCIZ-type integrals. Finally, we mention that HCIZ-type integrals have other applications in theoretical and mathematical physics beyond products of random matrices, e.g.~the hyperbolic HCIZ integral was used to find the spectral properties of the Wilson--Dirac operator in lattice quantum chromodynamics~\cite{KVZ13}. Moreover, HCIZ-type integrals represent a rich area of mathematical research, for example within the study of Lie groups, harmonic analysis, combinatorics, and probability (e.g. matrix-valued Brownian motion); see
e.g.~the text \cite{Te88a}.

\section{Products of random matrices and Hermite Muttalib--Borodin ensembles}
\label{sec:product}

In this section, we establish that the eigenvalue PDF of the matrix product~\eqref{W1}
is a polynomial ensemble and show that it reduces asymptotically to the Hermite Muttalib--Borodin ensemble~\eqref{MB-Hermite}.

As stated in the introduction, the eigenvalue PDF of~\eqref{W1} follows as a consequence of Theorem~\ref{T1}. The idea is simple: let $A$ be a random matrix from a polynomial ensemble, i.e. it has an eigenvalue PDF of the form
\begin{equation}\label{af}
P_A(\{a_k\}_{k=1}^n)=\frac1{Z_n}\prod_{1\leq i<j\leq n}(a_j-a_i)\det[w_j(a_i)]_{i,j=1}^n,
\end{equation}
where $a_1\leq a_{2}\leq \cdots\leq  a_{n}$ are the (ordered) eigenvalues of $A$, $w_j$ ($j=1,\ldots,n$) is a family of weight functions, and $Z_n$ is a normalisation constant. Now, let $G$ be an $n \times N$ $(n \le N)$ standard complex Gaussian matrix. Then Theorem~\ref{T1} gives the eigenvalue PDF of $G^\dagger AG$. Moreover, it is seen that this new eigenvalue PDF is also a polynomial ensemble. In other words, Theorem~\ref{T1} provides a map from the class of polynomial ensembles into itself. Thus, we may apply Theorem~\ref{T1} recursively to construct hierarchies of polynomial ensembles.
Let us make this statement more precise.

\begin{lemma}\label{C1}
Let $G$ be an $n \times N$ $(n \le N)$ standard complex Gaussian matrix, and let $A$ be a random matrix from a
polynomial ensemble with eigenvalue PDF (\ref{af}), independent of $G$.
Then the PDF for the non-zero eigenvalues of the random matrix product $G^\dagger AG$ is equal to
\begin{equation}\label{af2}
\frac1{Z_n}\prod_{l=1}^n\frac{1}{(N-l)!}\prod_{1\leq j<k\leq n}(x_k-x_j)
\det_{1\leq i,j\leq n}\bigg[\int_0^\infty \frac{da\,e^{-a}}{a^{n-N+1}}\,w_{j}\big(\frac{x_i}{a}\big)\bigg]
\end{equation}
with the eigenvalues ordered  $x_1\leq x_{2}\leq \cdots\leq  x_{n}$.
\end{lemma}

\begin{proof}
In order to use Theorem~\ref{T1}, we fix an $n_0\in\{0,1,\ldots,n\}$ and assume that the eigenvalues of $A$ are ordered as~\eqref{as}.
Consequently, the non-zero eigenvalues of $G^\dagger AG$  can be ordered as~\eqref{as2} almost surely.
It follows from the conditional eigenvalue PDF~\eqref{as1} and~\eqref{af} that the eigenvalue PDF of $G^\dagger AG$ (up to $N-n$ eigenvalues which are identically zero) is given by
\begin{multline}\label{recurrence-step}
\int_D P_A(\{a_k\}_{k=1}^n)P^{n_0}_{n}(\{a_j\}_{j=1}^n;\{x_j\}_{j=1}^n)\, da_1\cdots da_n
=\frac1{Z_n}\prod_{1\leq j<k\leq n}(x_k-x_j)\\
\times\int_D\prod_{l=1}^n\frac1{|a_l|}\frac{(x_l/a_l)^{N-n}}{(N-l)!}
\det
\begin{bmatrix}
 \{e^{-x_i/a_j}\}_{i,j=1}^{n_0} & 0 \\
 0 & \{e^{-x_{i}/a_{j}}\}_{i,j=n_0+1}^{n}
\end{bmatrix}
\det[w_j(a_i)]_{i,j=1}^n
\, da_1\cdots da_n,
\end{multline}
where the domain of integration $D$ is given according to~\eqref{as} and the eigenvalues $x_i$ ($i=1,\ldots,n$) are ordered according to~\eqref{as2}.
We note that the integral on the second line in~\eqref{recurrence-step} is a close cousin to the well-known Andreief integral~\cite{An83,Br55}.
Upon expansion of the determinants, it is readily seen that~\eqref{recurrence-step} may rewritten as
\begin{equation}\label{recurrence-step2}
\frac1{Z_n}\prod_{l=1}^n\frac{1}{(N-l)!}\prod_{1\leq j<k\leq n}(x_k-x_j)
\det
\begin{bmatrix}
 \displaystyle\bigg\{-\int_{-\infty}^0\frac{da}a (x_i/a)^{N-n} e^{-x_i/a}
 w_j(a)\bigg\}_{\substack{i=1,\ldots,n_0\\j=1,\ldots,n}} \\
 \displaystyle\bigg\{\int^{\infty}_0\frac{da}a (x_i/a)^{N-n} e^{-x_i/a}
 w_j(a)\bigg\}_{\substack{i=n_0+1,\ldots,n\\j=1,\ldots,n}}
\end{bmatrix}.
\end{equation}
If we make a change of variables $a\mapsto -a$ in the first $n_0$ rows in the determinant in~\eqref{recurrence-step2} then we get
\begin{equation}
\frac1{Z_n}\prod_{l=1}^n\frac{1}{(N-l)!}\prod_{1\leq j<k\leq n}(x_k-x_j)
\det_{1\leq i,j\leq n}\bigg[\int_0^\infty \frac{da}a\,\frac{e^{-|x_i|/a}}{(|{x_i}|/a)^{n-N}}\,w_j((\sgn x_i)a)\bigg];
\end{equation}
recall that $x_i$ ($i=1,\ldots,n$) is ordered according~\eqref{as2}. Finally, if we make another change of variables $a\mapsto |x_i|/a$ in the $i$-th row, then~\eqref{af2} follows for any fixed $n_0$. Note that this result is independent of the choice of $n_0$, so we have the final result.
\end{proof}

\begin{remark}
The study of maps from the space of polynomial ensembles onto itself is an interesting endeavour, since such maps give rise to new random matrix ensembles without destroying integrability. In fact, the study of such maps is presently an active area of research in random matrix theory~\cite{CKW15,Ku16,KK16,KR16}. Lemma~\ref{C1} provides a new transformation to this class of maps, which cannot be obtained directly from any of the previously established transformations. We note that Lemma~\ref{C1} includes the transformation~\cite[Theorem 2.1]{KS14} as a special case arising when the matrix $A$ is positive definite. A restriction of Lemma~\ref{C1} is that $n\leq N$. Thus it is seen that the PDF~\eqref{af2} develops a singularity for $n>N$ indicating that the formula is no longer generally valid in this case, depending on the properties of $w_j$. It would be interesting to extend the above results to include the case $n>N$ more generally.
\end{remark}

With Lemma~\ref{C1} at hand, we are ready to write down the eigenvalue PDF for the product~\eqref{W1}.

\begin{thm}\label{cor-matrix}  Let $\nu_0=0, \nu_1,  \ldots, \nu_M$ be non-negative integers.
Suppose that
$H$ is an $n\times n$ GUE matrix and  $G_1, \ldots, G_M$  are independent  standard complex Gaussian matrices where  $G_m$ is of size $(\nu_{m-1}+n) \times (\nu_{m} +n)$.
Then the joint PDF for the non-zero eigenvalues of the matrix~\eqref{W1} is given by
\begin{equation}\label{PDF-matrix}
P^{(M)}(x_1,\ldots,x_n)=\frac{1}{Z^{(M)}_n}\prod_{1\leq i<j\leq n}(x_j-x_i)\det_{1\leq i,j\leq n}\big[g_{j-1}^{(M)}(x_i)\big],
\end{equation}
where the weight functions $g_j^{(M)}$ are defined recursively by
\begin{equation}\label{weight-recursive}
g_j^{(0)}(x)=x^{j}e^{-x^2}
\quad\text{and}\quad
g_j^{(m)}(x)=\int_0^\infty \frac{dy}{y}\,y^{\nu_m}e^{-y}\,g_j^{(m-1)}(x/y),\quad m=1,\ldots,M
\end{equation}
and the normalisation constant is
\begin{equation}
Z^{(M)}_n={2^{-n(n-1)/2}\pi^{n/2}}\prod_{m=0}^M\prod_{j=1}^n\Gamma(\nu_m+j).
\end{equation}
\end{thm}

\begin{proof}
First, let us consider the simplest situation, that is a product of square matrices, i.e. $\nu_1=\cdots=\nu_M=0$. As the eigenvalue PDF of an $n\times n$ GUE matrix is given by~\eqref{PDF-matrix} with $M=0$, the theorem follows immediately by applying Lemma~\ref{C1} $M$ times.

We need to be a little more careful when the case of rectangular matrices is considered.
The $M=1$ case of the theorem is again an immediate consequence of Lemma~\ref{C1}, which gives us the eigenvalues of $W_1=G_1^\dagger HG_1$. However, in order to apply Lemma~\ref{C1} a second time and find the non-zero eigenvalues of $W_2=G_2^\dagger W_1G_2$, we have to take into account that $W_1$ has a zero eigenvalue with multiplicity $\nu_1$. To proceed, we can use the same idea as in~\cite{IK14}. The unitary invariance of Gaussian matrices tells us that $W_1\stackrel{d}=U^\dagger W_1 U$ and $G_2\stackrel{d}{=}VG_2$ for any $U,V\in U(n+\nu_1)$. It thus follows
\begin{equation}\label{reduction}
W_2=G_2^\dagger W_1 G_2\stackrel{d}{=}
\begin{bmatrix} \tilde G_2^\dagger & g_2^\dagger \end{bmatrix}
\begin{bmatrix} X_1 & 0 \\ 0 & 0 \end{bmatrix}
\begin{bmatrix} \tilde G_2 \\ g_2 \end{bmatrix}
=\tilde G_2^\dagger X_1 \tilde G_2,
\end{equation}
where $X_1=\diag(x_1,\ldots,x_n)$ is an $n\times n$ diagonal matrix distributed according to~\eqref{PDF-matrix} with $M=1$, while $\tilde G_2$ and $g_2$ are standard Gaussian matrices of size $n\times (n+\nu_2)$ and $\nu_1\times (n+\nu_2)$, respectively. Now, Lemma~\ref{C1} can be applied to the right-hand side in~\eqref{reduction}, which gives us the PDF of the non-zero eigenvalues of $W_2$. Repeating this procedure completes the proof.
\end{proof}

\begin{remark}
We note that the case $M=1$ of Theorem~\ref{cor-matrix} is in agreement with the result stated by Kumar~\cite[Eq.~(46) and~(47)]{Ku15}. However, the derivation therein is incomplete due to the reliance on the HCIZ integral
(1.14), rather than its hyperbolic variant (A.38) below.
\end{remark}

There are many other representations for the weight functions in Theorem~\ref{cor-matrix}
beyond the recursive definition~\eqref{weight-recursive}.
As usual, it is particularly useful for analytic purposes to write the weight functions in their contour integral representation.

\begin{lemma}
We have
\begin{equation}\label{contour-matrix}
g_j^{(M)}(x)=\frac{(\sgn x)^{j}}{4\pi i}\int_{c-i\infty}^{c+i\infty}ds\,|x|^s\,\Gamma\Big(\frac{j-s}{2}\Big)\prod_{m=1}^M\Gamma(\nu_m-s),
\end{equation}
where $c<0$ is a negative constant.
\end{lemma}

\begin{proof}
By means of the residue theorem, it is seen that
\begin{equation}
g_j^{(0)}(x)=x^{j}e^{-x^2}=\frac{(\sgn x)^{j}}{4\pi i}\int_{c-i\infty}^{c+i\infty}ds\,|x|^s\,\Gamma\Big(\frac{j-s}{2}\Big).
\end{equation}
Now, assume that $g_j^{(M-1)}(x)$ is given by~\eqref{contour-matrix}. From the recursive formula~\eqref{weight-recursive}, we have
\begin{equation}
g_j^{(M)}(x)=\int_0^\infty \frac{dy}{y}\,y^{\nu_M}e^{-y}\,\frac{(\sgn x)^{j}}{4\pi i}
\int_{c-i\infty}^{c+i\infty}ds\,\Big(\frac{|x|}y\Big)^s\,\Gamma\Big(\frac{j-s}{2}\Big)\prod_{m=1}^{M-1}\Gamma(\nu_m-s).
\end{equation}
It is a straightforward exercise, considering the asymptotic decay, to show that
with $c<0$ the order of the integrals may be interchanged.   Thus we have
\begin{equation}
g_j^{(M)}(x)=
\frac{(\sgn x)^{j}}{4\pi i}\int_{c-i\infty}^{c+i\infty}ds\,|x|^s\,\Gamma\Big(\frac{j-s}{2}\Big)\prod_{m=1}^{M-1}\Gamma(\nu_m-s)\,
\int_0^\infty \frac{dy}{y}\,y^{\nu_M-s}e^{-y}
\end{equation}
and the lemma follows by induction.
\end{proof}

As already mentioned, functions with a contour integral representation like~\eqref{contour-matrix} have certain properties which are useful for analytical purposes. 
In fact, many of these properties may be found in the literature if we first recognise the contour integral as a Fox $H$-function
\begin{equation}\label{weight-fox}
 g_j^{(M)}(x)=\frac{(\sgn x)^{j}}2
\FoxH{M+1}{0}{0}{M+1}{-}{(\nu_1,1),\ldots,(\nu_M,1),(\frac{j}2,\frac12)}{|x|}
\end{equation}
or as a Meijer $G$-function
\begin{equation}\label{weight}
 g_j^{(M)}(x)=(\sgn x)^{j}\prod_{m=1}^M\frac{2^{\nu_m-1}}{\sqrt{\pi}}
 \MeijerG{2M+1}{0}{0}{2M+1}{-}{\frac{\nu_1}{2},\frac{\nu_1+1}{2},\ldots,\frac{\nu_M}{2},\frac{\nu_M+1}{2},\frac{j}2}{\frac{x^2}{4^M}}.
\end{equation}
We refer to the book~\cite{MSH09} for an extensive review of these functions; the Fox $H$- and Meijer $G$-functions are defined by~\cite[Def.~1.1]{MSH09} and~\cite[Def.~1.5]{MSH09}, respectively.

As discussed in Section~\ref{sec:motivation}, one of our goals is to find a `classical gas' approximation for~\eqref{PDF-matrix}. For this purpose, we can use the asymptotic result~\cite{Fi72}
\begin{equation}\label{asymp}
\MeijerG{q}{0}{0}{q}{-}{b_1,\ldots,b_q}{x}\sim\frac1{q^{1/2}}\Big(\frac{2\pi}{x^{1/q}}\Big)^{(q-1)/2}
x^{(b_1+\cdots+b_q)/{q}}e^{-qx^{1/q}}\big(1+O(x^{1/q})\big),
\end{equation}
for $x\to\infty$ (recall that a typical eigenvalue grows with $n$). This immediately allows us to find a Muttalib--Borodin ensemble~\eqref{MB},
which approximates the product ensemble~\eqref{PDF-matrix}.
However, for notational simplicity, it is convenient to first make a change of variables
\begin{equation}\label{change}
x_i\mapsto x_i'=2^M\Big(\frac{y_i}{\sqrt{2M+1}}\Big)^{2M+1}
\end{equation}
for $i=1,\ldots,n$. Using the asymptotic formula~\eqref{asymp} and making a change of variable~\eqref{change}, we find the approximate PDF
\begin{equation}\label{hermite-MB}
\tilde P^{(M)}(y_1,\ldots,y_n)=\frac{1}{\tilde Z^{(M)}}\prod_{1\leq i<j\leq n} (y_j-y_i)(y_j^{2M+1}-y_i^{2M+1})
\prod_{k=1}^n|y_k|^\alpha e^{-y_k^2},
\end{equation}
where $\tilde Z^{(M)}$ is a new normalisation constant and $\alpha=\sum_{m=1}^M(2\nu_m+1)$.
We recognise~\eqref{hermite-MB} as the Hermite Muttalib--Borodin ensemble~\eqref{MB-Hermite}.
Note that the approximation~\eqref{asymp} is valid for large $x$ and that the absolute value of a typical eigenvalue grows with the matrix dimension $n$. Thus, one might suspect agreement between the two models in the large-$n$ limit except for local correlations near the origin. We will return to a comparison between the two models in Section~\ref{sec:hard} and Section~\ref{sec:global}.

It is worth noting that the following exact relation holds
\begin{equation}\label{asymp-exact}
\MeijerG{q}{0}{0}{q}{-}{0,\frac1q,\ldots,\frac{q-1}q}{x}=\frac{(2\pi)^{(q-1)/2}}{q^{1/2}}e^{-qx^{1/q}}
\end{equation}
for integer $q$.
This may be proven by writing the Meijer $G$-function on the left-hand side as its integral representation and then using Gauss' multiplication formula for the gamma functions. The exact relation~\eqref{asymp-exact} tells us that we can choose the parameters $b_k$ $(k=1,\ldots,q)$ in~\eqref{asymp} such that all subleading terms in the expansion vanish, and (\ref{hermite-MB}) is exact.

\section{Biorthogonality and correlations}
\label{sec:biortho}

Generally polynomial ensembles  describe determinantal point processes. The correlation kernel may be written as
\begin{equation}\label{kernel-sum}
K_n(x,y)=\sum_{k=0}^{n-1}\frac{p_k(x)\phi_k(x)}{h_k},
\end{equation}
where $p_k(x)$ and $\phi_k(x)$ are bi-orthogonal functions with normalisation $h_k$, i.e.
\begin{equation}
\int_{-\infty}^\infty dx\, p_k(x)\phi_k(x)=h_k\delta_{k\ell}.
\end{equation}
For both \eqref{PDF-matrix} and \eqref{hermite-MB} the $p_{k}(x)$ will be a monic polynomial of degree $k$, while
\begin{equation}
\phi_k(x)-g_k^{(M)}(x)\in\Span\{g_{k-1}^{(M)}(x),\ldots,g_{0}^{(M)}(x)\}
\end{equation}
for the product ensemble~\eqref{PDF-matrix} and
\begin{equation}
\phi_k(x)-x^{(2M+1)k} |x|^{\alpha}e^{-x^2}\in\Span\{x^{(2M+1)(k-1)} |x|^{\alpha}e^{-x^2},\ldots, |x|^{\alpha}e^{-x^2}\}
\end{equation}
for the Hermite Muttalib--Borodin ensemble~\eqref{hermite-MB}. In the latter case, the bi-orthogonal structure is already known~\cite{Ko67,Ca68,Bo98,FI16}. The main purpose of this section is to determine the bi-orthogonal functions --- and consequently the kernel~\eqref{kernel-sum} --- for the product ensemble~\eqref{PDF-matrix}.

\subsection{The oddness of being even}

For $M=0$, the bi-orthogonal functions are $p_n(x)=\tilde H_n(x)$ and $\phi_n(x)=\tilde H_n(x)e^{-x^2}$ with $\tilde H_n(x)$ denoting the Hermite polynomials in monic normalisation. We recall that the $n$-th Hermite polynomials is an even (odd) function when $n$ is even (odd);
this is due to the reflection symmetry of the Gaussian weight about the origin. A similar phenomenon is present for our product generalisation. In order to see this, we use an alternative form of the kernel~\eqref{kernel-sum}. We have~\cite{Co39,Bo98}
\begin{equation}\label{kernel}
K_n(x,y)=\sum_{k,\ell=0}^{n-1}(B_{n}^{(M)})^{-1}_{\ell,k}\,x^{k}g_{\ell}^{(M)}(y),
\end{equation}
where $B_n^{(M)}=(b_{i,j}^{(M)})_{i,j=0}^{n-1}$ is the $n$-th bi-moment matrix constructed from the bi-moments
\begin{equation}\label{bimoment}
b_{k,\ell}^{(M)}=\int_{-\infty}^\infty x^kg_\ell^{(M)}(x)dx.
\end{equation}
Simon~\cite{Si08} refers  the inverse moment matrix representation~\eqref{kernel} as the ABC (Aitken--Berg--Collar) theorem.

In the following, it will be useful to extend the concept of odd and even moments to bi-moments.
We say that the bi-moments are odd (even) when $k+\ell$ is odd (even).
Now, using that $g_\ell^{(M)}(-x)=(-1)^\ell g_\ell^{(M)}(x)$,
we see that the bi-moments satisfy
\begin{equation}
 b_{k,\ell}^{(M)}=(-1)^{k+\ell}b_{k,\ell}^{(M)},\qquad k,\ell=0,1,\ldots,
\end{equation}
which implies that all odd moments are equal to zero. In other words, the entries in the bi-moment matrix vanishes in a chequerboard pattern. Thus, by reordering rows and columns we may write the bi-moment matrix in a block diagonal form
$B_n^{(M)}\mapsto \diag(B_n^{(M,\text{even})},B_n^{(M,\text{odd})})$ with $B_n^{(M,\text{even})}=(b_{2k,2\ell}^{(M)})_{k,\ell}$ and
$B_n^{(M,\text{odd})}=(b_{2k+1,2\ell+1}^{(M)})_{k,\ell}$.
Using this reordering in the sum~\eqref{kernel}, we see that the kernel splits into two parts
\begin{equation}
K_n(x,y)=K_n^\text{even}(x,y)+K_n^\text{odd}(x,y),
\end{equation}
where
\begin{align}
K_n^\text{even}(x,y)&=\sum_{k,\ell=0}^{\lfloor \frac{n-1}2\rfloor}(B_n^{(M,\text{even})})^{-1}_{\ell,k}\,x^{2k}g_{2\ell}^{(M)}(y)
=\sum_{k=0}^{\lfloor \frac{n-1}2\rfloor} \frac{p_{2k}(x)\phi_{2k}(x)}{h_{2k}}, \label{kernel-even} \\
K_n^\text{odd}(x,y)&=\sum_{k,\ell=0}^{\lfloor \frac n2\rfloor-1}(B_n^{(M,\text{odd})})^{-1}_{\ell,k}\,x^{2k+1}g_{2\ell+1}^{(M)}(y)
=\sum_{k=0}^{\lfloor \frac n2\rfloor-1} \frac{p_{2k+1}(x)\phi_{2k+1}(x)}{h_{2k+1}}. \label{kernel-odd}
\end{align}
Here, the latter equality in both~\eqref{kernel-even} and~\eqref{kernel-odd} follow from comparison with~\eqref{kernel-sum}.
Finally, we note that
\begin{equation}
K_{2n}^\text{even}(x,y)=K_{2n-1}^\text{even}(x,y)
\qquad\text{and}\qquad
K_{2n}^\text{odd}(x,y)=K_{2n+1}^\text{odd}(x,y).
\end{equation}
Thus, in the following we can restrict our attention to kernels with an even subscript.

\subsection{Bi-orthogonal functions}

We are now ready to write down the bi-orthogonal functions for our product ensemble. As explained in the previous subsection, it is convenient to consider functions of odd and even degree separately.
\begin{prop}\label{thm:bi-func-sum}
The ensemble defined by Theorem~\ref{cor-matrix} is bi-orthogonalised by
\begin{align}
p_{2n}(x)&=\sum_{\ell=0}^n\frac{(-\tfrac14)^{n-\ell}}{(n-\ell)!}
\prod_{m=0}^M\frac{\Gamma(\nu_m+2n+1)}{\Gamma(\nu_m+2\ell+1)}x^{2\ell}, &
\phi_{2n}(x)&=\sum_{\ell=0}^n\frac{(-\tfrac14)^{n-\ell}}{(n-\ell)!}\frac{(2n)!}{(2\ell)!}g_{2\ell}^{(M)}(x), \nonumber \\
p_{2n+1}(x)&=\sum_{\ell=0}^n\frac{(-\tfrac14)^{n-\ell}}{(n-\ell)!}
\prod_{m=0}^M\frac{\Gamma(\nu_m+2n+2)}{\Gamma(\nu_m+2\ell+2)}x^{2\ell+1}, &
\phi_{2n+1}(x)&=\sum_{\ell=0}^n\frac{(-\tfrac14)^{n-\ell}}{(n-\ell)!}\frac{(2n+1)!}{(2\ell+1)!}g_{2\ell+1}^{(M)}(x), \nonumber \\
h_n&=2^{-n}\pi^{1/2}\prod_{m=0}^M\Gamma(\nu_m+n+1)\label{h_n},
\end{align}
with notation as above (recall that $\nu_0=0$).
\end{prop}

There are several different approaches to prove Proposition~\ref{thm:bi-func-sum}. Here, we will present a method which emphasizes the relation to the Hermite polynomials (see~\cite[Prop.~3.5]{Ip15} for the same method applied to the product ensemble of Laguerre type). In order to use this approach, it is convenient to first calculate the bi-moments.

\begin{lemma}
The bi-moments are given by
\begin{equation}\label{bimoment-gamma}
b_{k,\ell}^{(M)}=\Gamma\Big(\frac{k+\ell+1}{2}\Big)\prod_{m=1}^M\Gamma(\nu_m+k+1),
\end{equation}
for $k+\ell$ even and zero otherwise.
The bi-moment determinant is
\begin{equation}\label{bi-det}
D_n^{(M)}:=\det_{0\leq k,\ell\leq n}[b_{k\ell}^{(M)}]=2^{-n(n+1)/2}\pi^{(n+1)/2}\prod_{m=0}^M\prod_{j=0}^n\Gamma(\nu_m+j+1).
\end{equation}
\end{lemma}

\begin{proof}
We insert the contour integral representation of the weight functions~\eqref{contour-matrix} into the expression for the bi-moments~\eqref{bimoment}, then we see that the even moments are
\begin{equation}
b_{k\ell}=\frac{1}{2\pi i}\int_0^\infty dx\int_{c-i\infty}^{c+i\infty} ds \,x^s\,\Gamma\Big(\frac{k+\ell-s}{2}\Big)
\prod_{m=1}^M\Gamma(\nu_m+k-s).
\end{equation}
The integrals in this expression can be recognised as a combination of a Mellin and an inverse Mellin transform, which yields~\eqref{bimoment-gamma}.

In order to evaluate the bi-moment determinant~\eqref{bi-det}, we note that
\begin{equation}\label{bimoment-det}
D_n^{(M)}=\prod_{m=1}^M\prod_{j=0}^n\Gamma(\nu_m+j+1)\det_{0\leq k,\ell\leq n}[b_{k\ell}^{(M=0)}].
\end{equation}
This completes the proof, since the $M=0$ case is the well-known Hermite (or GUE) case.
\end{proof}

\begin{proof}[Proof of Proposition~\ref{thm:bi-func-sum}]
The bi-orthogonal functions may be expressed by means of their bi-moments exactly as orthogonal polynomials through their moments.
Thus, we have
\begin{equation}
p_n(x)=\frac{1}{D_{n-1}^{(M)}}
\det_{\substack{i=0,\ldots,n\\j=0,\ldots,n-1}}
\begin{bmatrix}
b_{i,j}^{(M)} \bigg\vert\, x^i
\end{bmatrix}
\qquad\text{and}\qquad
\phi_n(x)=\frac{1}{D_{n-1}^{(M)}}
\det_{\substack{i=0,\ldots,n-1\\j=0,\ldots,n}}
\bigg[\begin{array}{@{}cc @{}}
b_{i,j}^{(M)}  \\ \hline
g_j^{(M)}(x)
\end{array}\bigg].
\end{equation}
Furthermore, we have $h_n=D_n^{(M)}/D^{(M)}_{n-1}$.

The constants $h_n$ are immediate from the above and~\eqref{bi-det}. Thus, it remains only to find the bi-orthogonal functions. To do so, we first note that
\begin{align}
p_n(x)&=\frac{\prod_{k=0}^n\prod_{m=1}^M\Gamma(\nu_m+k+1)}{D_{n-1}^{(M)}}
\det_{\substack{i=0,\ldots,n\\j=0,\ldots,n-1}}
\begin{bmatrix}
b_{i,j}^{(M=0)} \bigg\vert\, \displaystyle\frac{x^i}{\prod_{m=1}^M\Gamma(\nu_m+i+1)}
\end{bmatrix},\\
\phi_n(x)&=\frac{\prod_{k=0}^{n-1}\prod_{m=1}^M\Gamma(\nu_m+k+1)}{D_{n-1}^{(M)}}
\det_{\substack{i=0,\ldots,n-1\\j=0,\ldots,n}}
\bigg[\begin{array}{@{}cc @{}}
b_{i,j}^{(M=0)}  \\ \hline
g_j^{(M)}(x)
\end{array}\bigg].
\end{align}
This observation is important, since we know that the monic Hermite polynomials (with respect to the weight $e^{-x^2}$) are given by
\begin{equation}
\tilde H_n(x)=2^{-n}H_n(x)=
\frac{1}{D_{n-1}^{(M=0)}}
\det_{\substack{i=0,\ldots,n\\j=0,\ldots,n-1}}
\begin{bmatrix}
b_{i,j}^{(M=0)} \bigg\vert\, x^i
\end{bmatrix}.
\end{equation}
It follows that the expressions for the bi-orthogonal function $p_n(x)$ and $\phi_n(x)$ can be found using the known expressions for the Hermite polynomials and then making substitutions
\begin{equation}
x^k\mapsto \frac{x^k}{\prod_{m=1}^M\Gamma(\nu_m+k+1)}
\qquad\text{and}\qquad
x^k\mapsto g_k^{(M)}(x),
\end{equation}
respectively. We recall that
\[
\tilde H_{2n}(x)=\sum_{\ell=0}^n\frac{(-1)^{n-\ell}}{(n-\ell)!}\frac{(2n)!}{(2\ell)!}\frac{(2x)^{2\ell}}{2^{2n}}
\quad\text{and}\quad
\tilde H_{2n+1}(x)=\sum_{\ell=0}^n\frac{(-1)^{n-\ell}}{(n-\ell)!}\frac{(2n+1)!}{(2\ell+1)!}\frac{(2x)^{2\ell+1}}{2^{2n+1}},
\]
which makes it a straightforward exercise to verify the proposition.
\end{proof}

\begin{remark}
We recall that the bi-orthogonal functions also can be obtained from the characteristic polynomial using that
\begin{equation}
p_N(x)=\big\langle\det[x\mathbb I_N-W_M]\,\big\rangle
\qquad\text{and}\qquad
\int_{\mathbb R} dx\,\frac{\phi_{N-1}(x)}{z-x}=\Big\langle\,\frac1{\det[z\mathbb I_N-W_M]}\,\Big\rangle
\end{equation}
with $\langle\cdots\rangle$ denoting the matrix average and $z\in\mathbb C\setminus\mathbb R$. The first relation allows for an alternative method to calculate $p_N(x)$; see e.g.~\cite{FL15}.
\end{remark}

\section{Integral representations and correlation kernels}
\label{sec:int}

The explicit expressions for the bi-orthogonal functions given by Proposition~\ref{thm:bi-func-sum} allow us to write down an explicit form for the correlation kernel by insertion in~\eqref{kernel-sum}. However, this formulation of the kernel is not optimal for asymptotic analysis. For this reason, in this section we will provide integral representations of the bi-orthogonal functions as well as the kernel.

\begin{prop}\label{prop:bi-func-int}
The bi-orthogonal functions given by Proposition~\ref{thm:bi-func-sum} have integral representations
\begin{align}
p_{2n}(|x|)&=\frac{\sqrt\pi(2n)!}{(-1)^n2^{2n}}\frac{1}{2\pi i}\oint_\Sigma ds\,|x|^{2s}\,
\frac{\Gamma(-s)}{\Gamma(n+1-s)\Gamma(s+\frac12)}
\prod_{m=1}^M\frac{\Gamma(\nu_m+2n+1)}{\Gamma(\nu_m+2s+1)}, \label{p2n-int} \\
p_{2n+1}(|x|)&=\frac{\sqrt\pi(2n+1)!}{(-1)^n2^{2n+1}}\frac{1}{2\pi i}\oint_\Sigma ds\,|x|^{2s+1}\,
\frac{\Gamma(-s)}{\Gamma(n+1-s)\Gamma(s+\frac32)}
\prod_{m=1}^M\frac{\Gamma(\nu_m+2n+2)}{\Gamma(\nu_m+2s+2)}, \label{p2n1-int} \\
\frac{\phi_{2n}(|x|)}{h_{2n}}&=\frac{(-1)^n\,2^{2n}}{\sqrt\pi(2n)!}\frac{1}{2\pi i}\int_{c-i\infty}^{c+i\infty} dt\,
|x|^{-2t-1}\frac{\Gamma(n-t)\Gamma(t+\frac12)}{\Gamma(-t)}
\prod_{m=1}^M\frac{\Gamma(\nu_m+2t+1)}{\Gamma(\nu_m+2n+1)}, \\
\frac{\phi_{2n+1}(|x|)}{h_{2n+1}}&=\frac{(-1)^n\,2^{2n+1}}{\sqrt\pi(2n+1)!}\frac{1}{2\pi i}\int_{c-i\infty}^{c+i\infty} dt\,
|x|^{-2t-2}\frac{\Gamma(n-t)\Gamma(t+\frac32)}{\Gamma(-t)}
\prod_{m=1}^M\frac{\Gamma(\nu_m+2t+2)}{\Gamma(\nu_m+2n+2)},
\end{align}
where the contour $\Sigma$ encloses the integers $0,1,\ldots,n$ in the negative direction and $-\frac12<c<0$.
We recall that $p_{2n}(x)$ and $\phi_{2n}(x)$ are even functions, while $p_{2n+1}(x)$ and $\phi_{2n+1}(x)$ are odd functions.
\end{prop}

\begin{proof}
The integrands in~\eqref{p2n-int} and~\eqref{p2n1-int} has $n+1$ simple poles located at $0,1,\ldots,n$, thus the series representations in Proposition~\ref{thm:bi-func-sum} follow upon a straightforward application of the residue theorem.

In order to find the integral representation of the bi-orthogonal functions $\phi_n(x)$, we first note that the weight functions can be written as
\begin{equation}\label{weight-int-proof}
g_\ell^{(M)}(x)=(\sgn x)^\ell\int_0^\infty \frac{dy}{y}\,y^\ell\,e^{-y^2}
\MeijerG{M}{0}{0}{M}{-}{\nu_1,\ldots,\nu_M}{\frac{|x|}{y}};
\end{equation}
this is easily seen starting from the recursive definition~\eqref{weight-recursive}. Now, using~\eqref{weight-int-proof} in the expression for $\phi_n(x)$ (cf. Proposition~\ref{thm:bi-func-sum}), we see that
\begin{equation}
\phi_n(x)=\int_0^\infty \frac{dy}{y}e^{-y^2}\tilde H_n(y)
\MeijerG{M}{0}{0}{M}{-}{\nu_1,\ldots,\nu_M}{\frac{|x|}{y}}.
\end{equation}
In other words, the bi-orthogonal functions $\phi_n(x)$ are an integral transform of the Hermite polynomials with respect to a Meijer $G$-function as integral kernel. The Hermite polynomial can itself be expressed as a Meijer $G$- or Fox $H$-function (see~\cite[Sec.~1.8.1.]{MSH09}) and the remaining integral is well-known from the literature~\cite[Sec.~2.3.]{MSH09}.
\end{proof}

\begin{prop}\label{prop:kernel-finite}
Integral representations of the even and odd kernels are
\begin{align}
 K_{2n}^\textup{even}(x,y)&=\frac{1}{2(2\pi i)^2}\int_{c-i\infty}^{c+i\infty} dt\oint_\Sigma ds\,\frac{|x|^{s}|y|^{-t-1}}{s-t}
 \frac{\Gamma(-\frac s2)\Gamma(\frac{t+1}2)}{\Gamma(-\frac t2)\Gamma(\frac{s+1}2)}\frac{\Gamma(\frac {2n-t}2)}{\Gamma(\frac{2n-s}2)}
 \prod_{m=1}^M\frac{\Gamma(\nu_m+t+1)}{\Gamma(\nu_m+s+1)}, \nonumber \\
 K_{2n}^\textup{odd}(x,y)&=\frac{\sgn(xy)}{2(2\pi i)^2}\int_{c-i\infty}^{c+i\infty} dt\oint_\Sigma ds\,\frac{|x|^{s}\,|y|^{-t-1}}{s-t}\,
 \frac{\Gamma(\frac{1-s}2)\Gamma(\frac{t+2}2)}{\Gamma(\frac{1-t}2)\Gamma(\frac{s+2}2)}
 \frac{\Gamma(\frac{2n-t+1}2)}{\Gamma(\frac{2n-s+1}2)}
 \prod_{m=1}^M\frac{\Gamma(\nu_m+t+1)}{\Gamma(\nu_m+s+1)},
\end{align}
with $-1<c<0$ and the contour $\Sigma$ is chosen such that it encircles $0,1,\ldots,2n-1$ in the negative direction with $\textup{Re}\{s\}>c$ for all $s\in\Sigma$.
\end{prop}

\begin{proof}
As the proofs for the odd and even kernels are almost identical, we provide only the proof for the even case. The odd case is easily verified by the reader.

It follows the definition of the even kernel~\eqref{kernel-even} together with contour integral representation of the bi-orthogonal functions from Proposition~\ref{prop:bi-func-int}, that
\begin{equation}
 K_{2n}^\text{even}(x,y)=\frac{1}{(2\pi i)^2}\int dt\oint_\Sigma ds\,|x|^{2s}|y|^{-2t-1}
 \frac{\Gamma(-s)\Gamma(t+\frac12)}{\Gamma(-t)\Gamma(s+\frac12)}\prod_{m=1}^M\frac{\Gamma(\nu_m+2t+1)}{\Gamma(\nu_m+2s+1)}
 \sum_{k=0}^{n-1}\frac{\Gamma(k-t)}{\Gamma(k+1-s)}.
\end{equation}
Following similar steps as in~\cite{KZ14}, we note the sum allows a telescopic evaluation. This gives
\begin{multline}
 K_{2n}^\text{even}(x,y)=\frac{1}{(2\pi i)^2}\int_{c-i\infty}^{c+i\infty} dt\oint_\Sigma ds\,\frac{|x|^{2s}|y|^{-2t-1}}{s-t}
 \frac{\Gamma(-s)\Gamma(t+\frac12)}{\Gamma(-t)\Gamma(s+\frac12)}\frac{\Gamma(n-t)}{\Gamma(n-s)}
 \prod_{m=1}^M\frac{\Gamma(\nu_m+2t+1)}{\Gamma(\nu_m+2s+1)}\\
 -\frac{1}{(2\pi i)^2}\int_{c-i\infty}^{c+i\infty} dt\oint_\Sigma ds\,\frac{|x|^{2s}|y|^{-2t-1}}{s-t}
 \frac{\Gamma(t+\frac12)}{\Gamma(s+\frac12)}\prod_{m=1}^M\frac{\Gamma(\nu_m+2t+1)}{\Gamma(\nu_m+2s+1)}.
\end{multline}
Here, the integrand on the second line is zero as it has no poles encircled by the contour $\Sigma$ and, thus
\begin{equation}
 K_{2n}^\text{even}(x,y)=\frac{1}{(2\pi i)^2}\int_{c-i\infty}^{c+i\infty} dt\oint_\Sigma ds\,\frac{|x|^{2s}|y|^{-2t-1}}{s-t}
 \frac{\Gamma(-s)\Gamma(t+\frac12)}{\Gamma(-t)\Gamma(s+\frac12)}\frac{\Gamma(n-t)}{\Gamma(n-s)}
 \prod_{m=1}^M\frac{\Gamma(\nu_m+2t+1)}{\Gamma(\nu_m+2s+1)}.
\end{equation}
Finally, the proposition follows by a change of variables $s\mapsto s/2$ and $t\mapsto t/2$.
\end{proof}


The above integral representations for the bi-orthogonal functions and the kernel are probably the most convenient form for further asymptotic analysis, as we will see in Section~\ref{sec:hard}. However, it is also often helpful to express these formulae in terms of special functions as (for example) it allows for use of pre-defined mathematical software. Furthermore, such reformulations often guide us to recognise patterns which otherwise would have been left unseen.

The integral representations for the bi-orthogonal functions given by Proposition~\ref{prop:bi-func-int} can also be recognised as several different types of special function; this includes generalised hypergeometric, Meijer $G$-, and Fox $H$-functions. Here, we will restrict ourselves to their Meijer $G$-function formulation.

Let us first consider the bi-orthogonal polynomials which may be written as
\begin{align}
p_{2n}(x)&=\frac{(-1)^n}{2^{2n}}\prod_{m=0}^M\frac{\Gamma(\nu_m+2n+1)}{2^{\,\nu_m}\pi^{-1/2}}  \nonumber
\MeijerG{1}{0}{1}{2M+2}{n+1}{-\frac{\nu_0}2,-\frac{\nu_0}2+\frac12,\ldots,-\frac{\nu_M}2,-\frac{\nu_M}2+\frac12}{\frac{x^2}{2^{2M}}}, \\
\frac{p_{2n+1}(x)}x&=\frac{(-1)^n}{2^{2n}}\prod_{m=0}^M\frac{\Gamma(\nu_m+2n+2)}{2^{\,\nu_m+1}\pi^{-1/2}}
\MeijerG{1}{0}{1}{2M+2}{n+1}{-\frac{\nu_0}2,-\frac{\nu_0}2-\frac12,\ldots,-\frac{\nu_M}2,-\frac{\nu_M}2-\frac12}{\frac{x^2}{2^{2M}}}.
\label{p_2n-meijer}
\end{align}
It is worth comparing these polynomials with the polynomial found in the study of the Laguerre-like matrix product~\eqref{old-product}.
Akemann et al.~\cite{AIK13} found that in this case the bi-orthogonal polynomial is given by
\begin{equation}
P_n^{(M)}(x)=(-1)^n\prod_{m=0}^M\Gamma(\nu_m+n+1)
\MeijerG{1}{0}{0}{M+1}{n+1}{-\nu_0,-\nu_1,\ldots,-\nu_M}{x}.
\end{equation}
It is clear that the two families of polynomials are related as
\begin{equation}
p_{2n}(x)\propto P_n^{(2M+1)}\Big(\frac{x^2}{2^{2M}}\Big)
\qquad\text{and}\qquad
p_{2n+1}(x)\propto xP_n^{(2M+1)}\Big(\frac{x^2}{2^{2M}}\Big)
\end{equation}
with
\begin{equation}\label{nu-map}
\{\nu_m\}_{m=0}^M\mapsto\{{\nu_m}/2,({\nu_m}-1)/2\}_{m=0}^M
\qquad\text{and}\qquad
\{\nu_m\}_{m=0}^M\mapsto\{{\nu_m}/2,({\nu_m}+1)/2\}_{m=0}^M,
\end{equation}
respectively.
This is a generalisation of the relation between Hermite and Laguerre polynomials. Recall that
\begin{equation}\label{HL}
\tilde H_{2n}(x)=\tilde L^{(-\frac12)}_n(x^2)
\qquad\text{and}\qquad
\tilde H_{2n+1}(x)=x\tilde L^{(+\frac12)}_n(x^2),
\end{equation}
where $\tilde H_n(x)$ and $\tilde L_n^{(\alpha)}(x)$ denote the Hermite and Laguerre polynomials in monic normalisation.

Likewise, the (non-polynomial) bi-orthogonal functions may be written as
\begin{align}
\frac{\phi_{2n}(|x|)}{|x|}&=(-1)^n\prod_{m=1}^M\frac{2^{\nu_m-2}}{\pi^{1/2}}
\MeijerG{2M+1}{1}{1}{2M+2}{-n}{\frac{\nu_M}2-\frac12,\frac{\nu_M}2,\ldots,\frac{\nu_0}2-\frac12,\frac{\nu_0}2}{\frac{x^2}{2^{2M}}}, \\
\phi_{2n+1}(|x|)&=(-1)^n\prod_{m=1}^M\frac{2^{\nu_m-1}}{\pi^{1/2}}
\MeijerG{2M+1}{1}{1}{2M+2}{-n}{\frac{\nu_M}2+\frac12,\frac{\nu_M}2,\ldots,\frac{\nu_0}2+\frac12,\frac{\nu_0}2}{\frac{x^2}{2^{2M}}}.
\end{align}
Again, we want to compare to the formula in~\cite{AIK13} which this time reads
\begin{equation}
\Phi_n^{(M)}(x)=(-1)^n\MeijerG{M}{1}{1}{M+1}{-n}{\nu_M,\ldots,\nu_1,\nu_0}{x}.
\end{equation}
Evidently, we have the following relations
\begin{equation}\label{phi-relations}
\phi_{2n}(|x|)\propto |x|\Phi_n^{(2M+1)}\Big(\frac{x^2}{2^{2M}}\Big)
\qquad\text{and}\qquad
\phi_{2n+1}(|x|)\propto \Phi_n^{(2M+1)}\Big(\frac{x^2}{2^{2M}}\Big),
\end{equation}
with~\eqref{nu-map} as before.
Yet again, this is a generalisation of the relation between Hermite and Laguerre polynomials. In the simplest case the relations~\eqref{phi-relations} reduces to
\begin{equation}
\tilde H_{2n}(x)w_\text{H}(x)=|x|\tilde L^{(-\frac12)}_n(x^2)w_\text{L}^{(-\frac12)}(x^2)
\qquad\text{and}\qquad
\tilde H_{2n+1}(|x|)w_\text{H}(x)=\tilde L^{(+\frac12)}_n(x^2)w_\text{L}^{(\frac12)}(x^2), \nonumber
\end{equation}
where $w_\text{H}(x)=e^{-x^2}$ and $w_\text{L}^{(\alpha)}(x)=x^\alpha e^{-x}$ are the Hermite and Laguerre weight functions.

It is, of course, well-known that there are relations between ensembles with reflection symmetry about the origin and ensembles on the half-line (albeit explicit formulae may be elusive). A general description of such relations in the Muttalib--Borodin ensembles can be found in~\cite{FI16}.

\section{Scaling limits at the origin in product and Muttalib--Borodin ensembles}
\label{sec:hard}

With the integral representations of the correlation kernels established by Proposition~\ref{prop:kernel-finite}, we can turn to a study of asymptotic properties. Perhaps the most interesting scaling regime is that of the local correlations near the origin, referred to as the hard edge when the eigenvalues are strictly positive. For other product ensembles~\cite{KZ14,Fo14,KKS15,FL16}, it has been observed that correlations at the hard edge is determined by the so-called Meijer $G$-kernel, which generalises the more familiar Bessel kernel. Below, we will see that the Meijer $G$-kernel  appears once again, but this time
involving a sum.

\begin{thm}\label{thm:hard} Let $K_n(x,y)=K_n^\textup{even}(x,y)+K_n^\textup{odd}(x,y)$ with the even and odd kernels given by Proposition~\ref{prop:kernel-finite}. For $x,y\in\mathbb R\setminus\{0\}$ and $\nu_1\ldots,\nu_M$ fixed, the microscopic limit near the origin is
\begin{equation}\label{hard-limit}
\lim_{n\to\infty}\frac{1}{\sqrt n}K_{2n}\Big(\frac x{\sqrt n},\frac y{\sqrt n}\Big)
=K^\textup{even}(x,y)+K^\textup{odd}(x,y)
\end{equation}
with
\begin{align}
K^\textup{even}(x,y)
&=\frac{1}{2(2\pi i)^2}\int_{c-i\infty}^{c+i\infty} dt\int_\Sigma ds\,\frac{|x|^{s}|y|^{-t-1}}{s-t}
 \frac{\Gamma(-\frac s2)\Gamma(\frac{t+1}2)}{\Gamma(-\frac t2)\Gamma(\frac{s+1}2)}
 \prod_{m=1}^M\frac{\Gamma(\nu_m+t+1)}{\Gamma(\nu_m+s+1)} \label{hard-even} \\
K^\textup{odd}(x,y)&=\frac{\sgn(xy)}{2(2\pi i)^2}\int_{c-i\infty}^{c+i\infty} dt\int_\Sigma ds\,\frac{|x|^{s}\,|y|^{-t-1}}{s-t}\,
 \frac{\Gamma(\frac{1-s}2)\Gamma(\frac{t+2}2)}{\Gamma(\frac{1-t}2)\Gamma(\frac{s+2}2)}
 \prod_{m=1}^M\frac{\Gamma(\nu_m+t+1)}{\Gamma(\nu_m+s+1)} \label{hard-odd},
\end{align}
where $-1<c<-1/2$ and $\Sigma$ encloses the non-negative half-line in the negative direction starting and ending at $+\infty$ such that $\textup{Re}\{s\}>c$ for all $s\in\Sigma$.
\end{thm}

\begin{proof}
We only consider the even kernel in Proposition~\ref{prop:kernel-finite} since the odd case is very similar.  After rescaling we rewrite  the integral representation of the even kernel in Proposition~\ref{prop:kernel-finite} as
\begin{align}\label{5.4}
\frac 1{\sqrt n} K_{2n}^\textup{even}(\frac x{\sqrt n},\frac y{\sqrt n})= \frac{1}{2(2\pi i)^2}\int_{c-i\infty}^{c+i\infty} dt\int_\Sigma ds\,\frac{|x|^{s}|y|^{-t-1}}{s-t} \frac{f_{n}(s)}{f_{n}(t)} \frac{g(t)}{g(s)},
\end{align}
 with \begin{equation}
 f_{n}(s)= \frac{\Gamma(n)\Gamma(-\frac s2)}{n^{\frac s2} \Gamma(n-\frac s2)}, \qquad 
 g(s)=\Gamma\big(\frac{s+1}{2}\big) \prod_{m=1}^M \Gamma(\nu_m+s+1). 
 \end{equation}
 
 For any fixed $t\in c+i \mathbb{R}$  and $s \in \Sigma$, using \cite[eq. 5.11.13]{NIST} we see 
\begin{equation}\label{5.6}
f_{n}(s)= \Gamma(-\frac s2) \big(1+O(\frac{1}{n})\big),  \qquad f_{n}(t)= \Gamma(-\frac t2) \big(1+O(\frac{1}{n})\big). \end{equation}
Formally, substituting (\ref{5.6}) in (\ref{5.4}) gives (\ref{hard-even}). To proceed rigorously, we need to verify   a condition for the  exchange of limit and integration.  For this purpose, we  will proceed to 
find two dominated functions respectively  corresponding to  $1/|f_{n}(t)|$  and  $|f_{n}(s)|$.

First,     using \cite[eq. 5.11.13]{NIST} we have for sufficiently  large $n$
\begin{equation}
 \frac{1}{|f_{n}(t)|} \leq  \frac{n^{\frac c2} \Gamma(n-\frac c2) }{ \Gamma(n)|\Gamma(-\frac t2)|} \leq   \frac{2}{|\Gamma(-\frac t2)|},   \quad \forall  t\in c+i \mathbb{R}. \label{t-bound}
\end{equation}

Second, we require an upper bound  for $ | f_{n}(s)  |$. Noting the asymptotic expansion, that as $z\rightarrow \infty$ in the sector $|\mathrm{arg}(z)|\leq \pi-\delta$ (with $0<\delta<\pi$)
\begin{equation}
\Gamma(z)=e^{-z}z^{z-\frac{1}{2}}\sqrt{2\pi} \big(1+O(\frac{1}{z})\big), \label{Agamma}
\end{equation} 
it is easy to see that for a given $y_0>0$ we can choose the contour $\Sigma=\Sigma_{l}\cup \Sigma_{r}$ with \begin{equation} \Sigma_{l}=\big\{\frac{c}{2}+iy:|y|\leq y_{0}\big\}\cup \big\{x\pm iy_0: \frac{c}{2} \leq x\leq 1\big\}, \quad \Sigma_{r}=  \big\{x\pm iy_0:   x> 1\big\}.\end{equation}
Thus, we get from \eqref{Agamma} and the boundedness of $\Gamma(-s/2)$ over $\Sigma_l$ that  for large $n$ there exists a constant $C_1=C_1(y_0) > 0$ such that 
\begin{equation}
 |f_{n}(s)| \leq C_1, \qquad \forall s\in  \Sigma_{l}. \label{upl}
\end{equation}
In order to estimate $f_{n}(s)$ with $s\in  \Sigma_{r}$, we use the integral representation  
 \begin{equation}
  f_{n}(s)=   \frac{n^{-\frac{s}{2}}}{ 2i \sin \frac{\pi s}{2}}    \int_{\mathcal{C}_0} (1-u)^{n-1}(-u)^{-\frac{s}{2}-1}du,
\end{equation}
where  $\mathcal{C}_0$ is a counter-clockwise path which begins and ends at $1$ and encircles the origin once; see e.g.  \cite[eq. 5.12.10]{NIST}. Note that we choose $(-u)^{-1-s/2}=e^{-(1+s/2)\log(-u)}$ with $-\pi<\mathrm{arg}(-u)<\pi$. Change $u$ by $u/n$ and deform the resulting contour into the path which starts from $n$, proceeds along the (upper) real axis to 1, describes a circle of radius one counter-clock round the origin and returns to $n$ along the (lower) real axis. That is, 
 \begin{equation}
  f_{n}(s)=   \frac{1}{ 2i \sin \frac{\pi s}{2}}    \int_{\mathcal{C}} (1-\frac{u}{n})^{n-1}(-u)^{-\frac{s}{2}-1}du.
\end{equation}
Let $s=v\pm iy_0, v>1$. On the unit circle of the $u$-integral above  write $-u=e^{i\theta}$. Then we easily obtain  for $n\geq1$
 \begin{equation}
  |f_{n}(s)|\leq    \frac{1}{ 2 |\sin \frac{\pi s}{2}|}    \int_{-\pi}^{\pi} \big(1+\frac{1}{n}\big)^{n-1} |e^{-(\frac{s}{2}+1)i\theta}|d\theta\leq 
  \frac{\pi e^{ 1+\frac{\pi y_{0}}{2}}}{   |\sin \frac{\pi s}{2}|}. \label{ub1}
\end{equation}
On the upper and lower  real axis,  we have 
 \begin{align}
  |f_{n}(s)|&\leq    \frac{1}{ 2 |\sin \frac{\pi s}{2}|}    \int_{1}^{n} \big(1-\frac{u}{n}\big)^{n-1}|u^{-\frac{s}{2}-1} e^{-(\frac{s}{2}+1)(\mp i\pi)}|du \nonumber \\
  &\leq \frac{1}{ 2 |\sin \frac{\pi s}{2}|}    \int_{1}^{n}   u^{-\frac{v}{2}-1} e^{\frac{1}{2}\pi y_{0}} du \nonumber\\
 & =\frac{1}{  |\sin \frac{\pi s}{2}|}  e^{\frac{1}{2}\pi y_{0}}    \frac{1- n^{-\frac{v}{2}} }{v}\leq \frac{1}{  |\sin \frac{\pi s}{2}|}  e^{\frac{1}{2}\pi y_{0}}.  \label{ub2}
\end{align} 
Using the simple fact $|\sin \frac{\pi s}{2}| \geq |\sinh \frac{\pi }{2} \mathrm{Im}(s)|$, combination of \eqref{ub1} and \eqref{ub2} shows that 
there exists a constant $C_2=C_2(y_0)>0$ such that 
\begin{equation}
 |f_{n}(s)| \leq C_2, \qquad \forall s\in  \Sigma_{r}.
\end{equation}
Together with \eqref{upl} this gives us a bound $C>0$, that is, for large $n$  
\begin{equation}
 |f_{n}(s)| \leq C, \qquad \forall s\in  \Sigma. \label{s-bound}
\end{equation}

Finally, use \eqref{Agamma} and  the  asymptotic formula  that as $y\rightarrow \pm \infty$ 
\begin{equation}
| \Gamma(x + iy) |\sim  \sqrt{2\pi} |y|^{x-\frac{1}{2}}  e^{-\frac{1}{2}\pi |y|}
\end{equation}
with bounded real value of $x$ (see \cite[eq. 5.11.9]{NIST}),  it is easy to conclude  that the  function  of  variables $s$ and $t$
 \begin{equation}
 \,\frac{||x|^{s}|y|^{-t-1}|}{|s-t|}   \frac{2}{|\Gamma(-\frac t2)|}   \frac{|g(t)|}{|g(s)|},  
 \end{equation}
  is integrable along the chosen contours, 
 whenever $-1<c<-1/2$. Here  we emphasize that the assumption $-1/2\leq c<0$  does  not ensure the convergence in the special case $M=0$ while for  $M\geq 1$ it can be relaxed to $-1<c<0$ as in Proposition~\ref{prop:kernel-finite}.
 
With this, combing  \eqref{t-bound} and \eqref{s-bound}, we have indeed justified  the interchange of limit and integrals for every $M$ by the dominated convergence theorem, which completes the proof.
\end{proof}


For a comparison with other known results, it is useful to rewrite the hard edge correlation function of Theorem~\ref{thm:hard} in terms of Meijer $G$-functions.
Using that
\begin{equation}
\int_0^1du\,u^{s-t-1}=\frac1{s-t},
\end{equation}
we see that the even and odd kernel can be written as
\begin{align}
K^\textup{even}(|x|,|y|)=\frac{|y|}{2^{2M}}\int_0^1du\, \nonumber
&\MeijerG{1}{0}{0}{2M+2}{-}{-\frac{\nu_0}2,-\frac{\nu_0}2+\frac12,\ldots,-\frac{\nu_M}2,-\frac{\nu_M}2+\frac12}{\frac{x^2}{2^{2M}}u} \\
&\times\MeijerG{2M+1}{0}{0}{2M+2}{-}{\frac{\nu_M}2-\frac12,\frac{\nu_M}2,\ldots,\frac{\nu_0}2-\frac12,\frac{\nu_0}2}{\frac{y^2}{2^{2M}}u}, \\
K^\textup{odd}(|x|,|y|)=\frac{|x|}{2^{2M}}\int_0^1du\, \nonumber
&\MeijerG{1}{0}{0}{2M+2}{-}{-\frac{\nu_0}2,-\frac{\nu_0}2-\frac12,\ldots,-\frac{\nu_M}2,-\frac{\nu_M}2-\frac12}{\frac{x^2}{2^{2M}}u} \\
&\times\MeijerG{2M+1}{0}{0}{2M+2}{-}{\frac{\nu_M}2+\frac12,\frac{\nu_M}2,\ldots,\frac{\nu_0}2+\frac12,\frac{\nu_0}2}{\frac{y^2}{2^{2M}}u},
\end{align}
respectively. We recall that the so-called Meijer $G$-kernel is given by~\cite{KZ14}
\begin{equation}
K_\text{Meijer}^M(x,y)=\int_0^1du\,\MeijerG{1}{0}{0}{M+1}{-}{-\nu_0,\ldots,-\nu_M}{xu}\MeijerG{M}{0}{0}{M+1}{-}{\nu_M,\ldots,\nu_0}{yu} \label{G-kernel}
\end{equation}
with $x,y>0$. We note that this kernel is single-sided ($x,y\in\mathbb R_+$) while the kernel from Theorem~\ref{thm:hard} is double-sided ($x,y\in\mathbb R\setminus\{0\}$). However, it is also evident that our new kernel may be re-expressed in terms of the Meijer $G$-kernel. We have
\begin{equation}
K^\textup{even}(|x|,|y|)=\frac{|y|}{2^{2M}}K_\text{Meijer}^{2M+1}\Big(\frac{x^2}{2^{2M}},\frac{y^2}{2^{2M}}\Big)
\quad\text{and}\quad
K^\textup{odd}(|x|,|y|)=\frac{|x|}{2^{2M}}K_\text{Meijer}^{2M+1}\Big(\frac{x^2}{2^{2M}},\frac{y^2}{2^{2M}}\Big)
\end{equation}
with
\begin{equation}
\{\nu_m\}_{m=0}^M\mapsto\{{\nu_m}/2,({\nu_m}-1)/2\}_{m=0}^M
\qquad\text{and}\qquad
\{\nu_m\}_{m=0}^M\mapsto\{{\nu_m}/2,({\nu_m}+1)/2\}_{m=0}^M,
\end{equation}
respectively. Thus, the random product matrix~\eqref{W1} provides yet another appearance of the Meijer $G$-kernel; albeit this time in a double-sided version. For graphical representation of the Meijer $G$-kernel we refer to~\cite[Fig.~3.2]{Ip15}, which shows plots of the local density (i.e. the kernel with $x=y$) for different values of $M$.

A double-side hard edge scaling limit near the origin is also present in the Hermite Muttalib--Borodin ensemble.
In this case the kernel is found to be~\cite{Bo98}
\begin{equation}\label{kernel-borodin}
K^\text{even}(x,y)=K^{(\frac{\alpha-1}2,\theta)}(x^2,y^2)
\qquad\text{and}\qquad
K^\text{odd}(x,y)=\sgn(xy)|x|^\theta|y|\,K^{(\frac{\alpha+\theta}2,\theta)}(x^2,y^2)
\end{equation}
where
\begin{equation}\label{kernel-wright-bessel}
K^{(\alpha,\theta)}(x,y)=\theta\int_0^1du(xu)^\alpha J_{\frac{\alpha+1}\theta,\frac1\theta}(xu)J_{\alpha+1,\theta}((yu)^\theta)
\end{equation}
with $J_{a,b}(x)$ denoting Wright's Bessel function. In the case relevant to us~\eqref{MB-Hermite}, we also have $\theta=2M+1$. Furthermore, it is known from~\cite{KS14} that the kernel~\eqref{kernel-wright-bessel} is a Meijer $G$-kernel whenever $\theta$ is a positive integer. In particular, we have
\begin{equation}
\Big(\frac{x^2}{2^{2M}}\Big)^{\frac{1}{2M+1}-1}
K^{(\alpha,2M+1)}\Big((2M+1)\Big(\frac{x^2}{2^{2M}}\Big)^{\frac{1}{2M+1}},(2M+1)\Big(\frac{y^2}{2^{2M}}\Big)^{\frac{1}{2M+1}}\Big)
=K_\text{Meijer}^{2M+1}\Big(\frac{y^2}{2^{2M}},\frac{x^2}{2^{2M}}\Big),
\end{equation}
where the Meijer $G$-kernel on the right-hand side has indices
\begin{equation}
\nu_m=\frac{\alpha+m-1}{2M+1},\qquad m=1,\ldots,2M+1,
\end{equation}
and as always $\nu_0=0$. It follows from~\eqref{kernel-borodin} and~\eqref{kernel-wright-bessel} that the hard edge correlations for the Hermite Muttalib--Borodin ensemble with appropriately chosen parameters may be expressed in terms of the Meijer $G$-kernel in a similar fashion as done for the product ensemble above. We note that the choice of variables in~\eqref{kernel-wright-bessel} should be compared to the change of variables~\eqref{change} performed in the derivation of the asymptotic reduction~\eqref{hermite-MB}.

It is worth verifying consistency of the simplest scenario of $M=0$.
When  $M=0$ our matrix ensemble~\eqref{W1} reduces to the GUE, hence the kernel given by Theorem~\ref{thm:hard} must reduce to the sine kernel for $M=0$. To see this, we use 
\begin{equation}
\MeijerG{1}{0}{0}{2}{-}{0,\frac12}{\frac{x^2}{4}}=\frac{\cos x}{\sqrt\pi}
\qquad\text{and}\qquad
\MeijerG{1}{0}{0}{2}{-}{\frac12,0}{\frac{x^2}{4}}=\frac{\sin |x|}{\sqrt\pi}.
\end{equation}
It follows that
\begin{align}
K^\textup{even}(x,y)&=\frac{1}{\pi}\int_0^1\frac{du}{\sqrt u}\cos(2x\sqrt u)\cos(2y\sqrt u)
=\frac1{\pi}\Big(\frac{\sin 2(x-y)}{2(x-y)}+\frac{\sin 2(x+y)}{2(x+y)}\Big), \\
K^\textup{odd}(x,y)&=\frac{1}{\pi}\int_0^1\frac{du}{\sqrt u}\sin(2x\sqrt u)\,\sin(2y\sqrt u)
=\frac1{\pi}\Big(\frac{\sin 2(x-y)}{2(x-y)}-\frac{\sin 2(x+y)}{2(x+y)}\Big),
\end{align}
which upon insertion into~\eqref{hard-limit} indeed reproduces the sine kernel.

In the end of this section, let us emphasize that there  also exists a contour integral  representation of the limiting kernel in Theorem~\ref{thm:hard} which combines the odd and even into a single formula.

\begin{prop}\label{prop:kernelrep2} With the same  notation as  in Theorem \ref{thm:hard}, the limiting kernel at the origin can be rewritten as
\begin{align}
 K^\textup{even}(x,y)+ K^\textup{odd}(x,y)=2\, \mathcal{K}_{\nu_{1},\ldots,\nu_{M}}(2x,2y), \label{equivalence}
 \end{align}
where the kernel on the right-hand side is defined as
\begin{align}
\mathcal{K}_{\nu_{1},\ldots,\nu_{M}}(x,y)&
=\int_{C_{R}} \frac{dv}{2\pi i}
\,\MeijerG{1}{0}{0}{M+1}{-}{0,-\nu_1, \ldots,-\nu_M}{-\sgn(y)xv}\MeijerG{M+1}{0}{0}{M+1}{-}{0, \nu_1,\ldots,\nu_M}{|y|v},
\label{doubleG-kernel}
\end{align}
with $C_{R}$ denoting a path in the right-half plane from $-i$ to $i$.
\end{prop}

\begin{proof}
Using  Euler's reflection formula and duplication formula for the gamma function, we see that
\begin{equation*}
K^\textup{even}(x,y)+ K^\textup{odd}(x,y)=\frac{1}{(2\pi i)^2}\int dt\int ds\, (2|x|)^{s}(2|y|)^{-t-1}
\frac{ g(s,t)}{s-t} \frac{\Gamma(t+1)}{\Gamma(s+1)}
 \prod_{m=1}^M\frac{\Gamma(\nu_m+t+1)}{\Gamma(\nu_m+s+1)},
\end{equation*}
where
\begin{equation}
g(s,t)=\frac{\sin\frac{\pi}{2}t}{\sin\frac{\pi}{2}s}+\sgn(xy) \frac{\cos\frac{\pi}{2}t}{\cos\frac{\pi}{2}s}.
\end{equation}
In order to proceed, we will consider the cases $xy<0$ and $xy>0$ separately.
For $xy<0$,  it is seen that
\begin{equation}
g(s,t)= \frac{2}{\sin\pi s}\sin\frac{\pi}{2}(t-s)=-\frac{2}{\pi} \Gamma(-s)\Gamma(1+s) \, \sin\frac{\pi}{2}(t-s).
\end{equation}
Now~\eqref{equivalence} can be obtained using the integral representation
\begin{equation}
 \frac{1}{\pi i} \int_{C_{R}}dv \,v^{s-t-1}= \frac{1}{t-s}\sin\frac{\pi}{2}(t-s),
\end{equation}
with the contour $C_R$ as above,
together with the definition of Meijer $G$-function. For $xy>0$, we note that
\begin{equation}
e^{i\pi s}g(s,t)=
\left(\frac{\sin\frac{\pi}{2}t}{\sin\frac{\pi}{2}s}-\frac{\cos\frac{\pi}{2}t}{\cos\frac{\pi}{2}s} \right)+2e^{i\frac{\pi}{2}(t+s)}.
\end{equation}
The $s$-variable integrand in the second part has no pole within the contour $\Sigma$. Thus, the problem reduces to the proven situation.
\end{proof}

The simplest non-trivial case is $M=1$. Here, we get
 \begin{equation}
\mathcal{K}_{\nu}(x,y)
=\left(\frac{y}{x}\right)^{\nu/2}  \frac{1}{\pi i} \int_{C_{R}}dv
\, I_{\nu}\big(2\sqrt{\sgn(y)xv}\big)\,  K_{\nu}\big(2\sqrt{|y|v}\big),  \label{doubleM1-kernel}
\end{equation}
with the modified Bessel functions $I_{\nu}$ and $K_{\nu}$, which follows immediately from the fact that
 \begin{align}
\MeijerG{1}{0}{0}{2}{-}{0,-\nu}{-z}=z^{-\nu/2} I_{\nu}(2\sqrt{z}), \qquad  \MeijerG{2}{0}{0}{2}{-}{\nu,0}{z}=2 z^{\nu/2} K_{\nu}(2\sqrt{z}).
\end{align}

\section{Global spectra in product and Muttalib--Borodin ensembles}
\label{sec:global}

The study of the scaling limit  at  the origin  in the previous section introduces a scale in which the average spacing between eigenvalues is of order unity. A very different, but still well-defined, limiting process is the so-called global scaling regime. In this regime the average spacing between eigenvalues tends to zero in such way that the spectral density tends to a quantity $\rho(x)$ with compact support $I\subset\mathbb R$ and $\int_I \rho(x)dx=1$. Here $\rho(x)$ is referred to as the global density.
Throughout this section, the indices $\nu_1\ldots,\nu_M$ are kept fixed.

For the Laguerre Muttalib--Borodin ensemble specified by the density~\eqref{MB-laguerre} the global scaling limit corresponds to a change of variables $x_j\mapsto nx_j$. Introducing the further change of variables $x_j\mapsto Mx_j^M$, the global density is known to be the so-called Fuss--Catalan density with parameter $M$ \cite{FW15}. It can be specified by the moment sequence
\begin{equation}
\text{FC}_M(k)=\frac1{Mk+1}\binom{(M+1)k}k,\qquad k=0,1,\ldots\,.
\end{equation}
These are the Fuss--Catalan numbers (the Catalan numbers are the case $M=1$).

Now, consider the product of $M$ standard complex Gaussian random matrices. Consistent with the discussion in Section~\ref{sec:motivation}, the corresponding global density is again the Fuss--Catalan density with parameter $M$~\cite{Mu02,AGT10,BBCC11,NS06}.

It is known that the Fuss--Catalan density, $\rho^{(M)}_\text{FC}(x)$ say, can also be characterised as the minimiser of the energy functional
\begin{equation}\label{energy-laguerre}
E[\rho]=M\int_0^Ldx\,\rho(x)x^{\frac1M}-\frac{1}{2}\int_0^Ldx\int_0^Ldy\,\rho(x)\rho(y)\log\big(|x-y||x^{\frac1M}-y^{\frac1M}|\big)
\end{equation}
with $L=(M+1)^{M+1}/M^M$; see~\cite{CR14,FL15,FLZ15}. Note that the energy functional~\eqref{energy-laguerre} relates to~\eqref{MB-laguerre} through the aforementioned change of variables. Similarly, the energy functional corresponding to~\eqref{MB-Hermite} is
\begin{align}
\tilde E[\tilde\rho]&=\theta\int_{-\tilde L}^{\tilde L}dx\,\rho(x)x^{\frac2\theta}
-\frac{1}{2}\int_{-\tilde L}^{\tilde L}dx\int_{-\tilde L}^{\tilde L}dy\,
\tilde\rho(x)\tilde\rho(y)\log\big(|x-y||\sgn x|x|^{\frac1\theta}-\sgn y|y|^{\frac1\theta}|\big) \nonumber \\
&=2\theta\int_{0}^{\tilde L}dx\,\rho(x)x^{\frac2\theta}-\int_{0}^{\tilde L}dx\int_{0}^{\tilde L}dy\,
\tilde\rho(x)\tilde\rho(y)\log\big(|x^2-y^2||(x^2)^{\frac1\theta}-(y^2)^{\frac1\theta}|\big)
\label{energy-hermite}
\end{align}
with $\theta=2M+1$.
We note that changing variables $x^2\mapsto x$ and $y^2\mapsto y$, then setting $\tilde\rho(x)=x\rho(x^2)$ reduces~\eqref{energy-hermite} to~\eqref{energy-laguerre} with $L=\tilde L^2$. Thus, the minimiser in~\eqref{energy-hermite} is given in terms of the Fuss-Catalan density
\begin{equation}\label{double-sided-FC}
\tilde\rho(x)=|x|\rho_\text{FC}^{(M)}(x^2)
\end{equation}
and is symmetric about the origin.

As an illustration, let us consider the simplest case, $M=1$. The Fuss--Catalan density becomes the celebrated Mar\v cenko--Pastur density,
\begin{equation}
\rho_\text{FC}^{(M=1)}(x)=\frac{1}{2\pi}\sqrt{\frac{4-x}{x}},\qquad 0<x<4.
\end{equation}
The formula~\eqref{double-sided-FC} then gives the standard result (see e.g.~\cite{PSbook}) that the energy functional
\begin{equation}
 \tilde E[\tilde\rho]=\int_{-2}^{2}dx\,\tilde \rho(x)x^{2}
-\int_{-2}^{2}dx\int_{-2}^{2}dy\,\tilde\rho(x)\tilde\rho(y)\log|x-y|
\end{equation}
is minimised by
\begin{equation}
\rho_\text{Wigner}(x)=\frac{\sqrt{4-x^2}}{2\pi},\qquad -2<x<2,
\end{equation}
which is Wigner's semi-circle law.

It has been demonstrated in Section~\ref{sec:product}, that the energy function implicit in~\eqref{energy-hermite} underlies the eigenvalue distribution of the random matrix product~\eqref{W1}. Thus, we can anticipate that after appropriate scaling the global density for the product ensembles is given by~\eqref{double-sided-FC}. A direct proof of this can obtained through a number of different strategies. We consider first a method based on the characteristic polynomial.

In terms of the global scaled variables, the key equation relating the averaged characteristic polynomial to  the global is the asymptotic formula~\cite{FW15}
\begin{equation}\label{asymp-stieltjes}
\frac1n\frac d{dz}\log\big\langle\det(z n^{M+\frac12}\mathbb I_n-W_M)\big\rangle=\tilde G_M(z)+O(n^{-1}),
\end{equation}
where $\tilde G_M(z)$ is the Stieltjes transform of the global spectral density,
\begin{equation}\label{stieltjes}
\tilde G_M(z)=\int_{-\tilde L}^{\tilde L}dx\,\frac{\tilde\rho(x)}{z-x}.
\end{equation}
Following the strategy first used in~\cite{FL15}, the formula~\eqref{asymp-stieltjes} leads to a characterisation of the Stieltjes transform~\eqref{stieltjes}, upon realising that the characteristic polynomial satisfy a linear differential equation.

\begin{prop}
Consider a matrix product~\eqref{W1} with even matrix dimension, $n=2N$. Let
\begin{equation}\label{char-poly}
f(z)=\big\langle\det(z\mathbb I_{n}-W_M)\big\rangle
\end{equation}
denote the characteristic polynomial. Then~\eqref{char-poly} is a solution to the $(2M+2)$-th differential equation,
\begin{equation}\label{diff}
2z^2\Big(z\frac d{dz}-n\Big)f(z)
=\prod_{m=0}^M\Big(z\frac d{dz}+\nu_m\Big)\Big(z\frac d{dz}+\nu_m-1\Big)f(z),
\end{equation}
with asymptotic boundary condition $f(z)\sim z^n$ for $|z|\to \infty$.
\end{prop}

\begin{proof}
The characteristic polynomial~\eqref{char-poly} is identical to the bi-orthogonal polynomial $p_{2N}(z)$. As shown earlier, this polynomial is proportional to a Meijer $G$-function~\eqref{p_2n-meijer}. It is well-known that such Meijer $G$-functions satisfy the differential equation~\eqref{diff}. The asymptotic boundary condition follows trivially, since $f(z)$ is a monic polynomial.
\end{proof}

Changing variables $z\mapsto n^{M+\frac12}\hat z/\sqrt{2}$ in~\eqref{diff} and using that~\cite{FL15}
\begin{equation}
\frac{f^{(k)}(\hat z)}{f(\hat z)}\sim\Big(\frac{f'(\hat z)}{f(\hat z)}\Big)^k
\end{equation}
to leading order in $n$, we see that for large $n$ the differential equation~\eqref{diff} reduces to the algebraic equation (see e.g. \cite{bai07} for $M=1$)
\begin{equation}\label{alg-eq1}
z^2(z\tilde G_M(z)-1)=(z\tilde G_M(z))^{2M+2}
\end{equation}
with asymptotic condition $\tilde G_M(z)\sim 1/z$ as $|z|\to\infty$.
This equation is to be compared to the algebraic equation satisfied by the Stieltjes transform of the Fuss--Catalan density,
\begin{equation}\label{alg-eq2}
z(zG_M(z)-1)=(zG_M(z))^{M+1},
\end{equation}
see e.g.~\cite{FL15}. With $z\mapsto z^2$ and $M\mapsto 2M+1$ and setting $\tilde G_{M}(z)=zG_{2M+1}(z^2)$, we see that~\eqref{alg-eq2} reduces to~\eqref{alg-eq1}. This prescription is equivalent to ~\eqref{double-sided-FC}, thus verifying this formula as the evaluation of the global density.

The same result can also be obtained using free probability techniques. To see this, we need some additional notation. Let $a$ be a non-commutative random variable with distribution $d\mu(x)=\rho(x)dx$. The Stieltjes transform $G_a(z)$ of the variable $a$ is defined analogous to~\eqref{stieltjes}. The $S$-transform is defined as
\begin{equation}
S_a(z)=\frac{1+z}{z}\gamma^{-1}(z)
\qquad\text{with}\qquad
\gamma(z)=-1+z^{-1}G_a(z^{-1}).
\end{equation}
Now assume that $a$ and $b$ are two freely independent non-commutative random variables and that the Stieltjes transform $G_b(z)$ satisfies a functional equation $P(z,G_b(z))=0$.
It is known~\cite{NS06} that under these conditions the Stieltjes transform $G_{ab}(z)$ of the product $ab$ satisfies
\begin{equation}\label{functional-recursion}
P\Big(zS_a(zG_{ab}(z)-1),\frac{zG_{ab}(z)}{S_a(zG_{ab}(z)-1)}\Big)=0.
\end{equation}
Moreover, we know that if $a$ is given by the free normal distribution (i.e. Wigner's semi-circle) and $b$ is given by the free Poisson distribution (i.e. Mar\v cenko--Pastur), then
\begin{equation}
S_a(z)=\frac1{1+z} \qquad\text{and}\qquad G_{b}(z)^2-zG_b(z)+1=0.
\end{equation}
We can now use that the limiting distributions for the GUE and the Wishart ensemble are the free normal and the free Poisson, respectively. Thus, using~\eqref{functional-recursion} $M$ times, we see that our product~\eqref{W1} indeed gives rise to the the functional equation~\eqref{alg-eq1}.

It is also possible to construct a parametrisation of the global density in terms of elementary functions based on the polynomial equation \eqref{alg-eq1}. With
\begin{equation}
x_{0}^{2}=\frac{\big(\sin((2M+2)\varphi)\big)^{2M+2}}{\sin\varphi\,\big(\sin((2M+1)\varphi)\big)^{2M+1}},
\qquad  0\leq\varphi\leq\frac{\pi}{2M+2},
\end{equation}
we have
\begin{equation}
\tilde\rho(x_0)
=\frac{1}{\pi}\sqrt{\frac{\sin\varphi}{\sin(2M+ 1)\varphi}}\left(\frac{\sin(2r+1)\varphi}{\sin(2M+ 2)\varphi}\right)^{ M}\sin\varphi,
\qquad  0\leq\varphi\leq\frac{\pi}{2M+2};
\end{equation}
see e.g.~\cite{FL15}.  We remark that it follows that the singularity at the origin blows up like
\begin{equation}
\tilde\rho(x_0)\sim  \frac{1}{\pi } \sin\frac{\pi}{2M+2} \, |x_0|^{-\frac{M}{M+1}}
\end{equation}
as $x_0\to0$.

\section{Conclusion and outlook}

In this paper, we have shown that it is possible to construct a Hermitised random matrix product for which the eigenvalues form a determinantal point process  on the entire real line with an explicit kernel. This is a fundamental new contribution to the study of random matrix product ensembles, since all previous exactly solvable models of this type have had eigenvalues restricted to the positive half-line. Furthermore, we have argued that this Hermitised product ensemble can be considered a natural generalisation of the classical Hermite ensemble (i.e. GUE) in similar way as the squared singular values of matrix products with  Gaussian matrices~\cite{AKW13,AIK13} and truncated unitary matrices~\cite{KKS15} can be considered generalisations of the Laguerre and Jacobi ensembles, respectively.
To this point, we have shown that the joint eigenvalue PDF reduces asymptotically to the Muttalib--Borodin ensemble of Hermite type.

On another front, we have shown that the local scaling limit near the origin is described by a two-sided generalisation of the so-called Meijer $G$-kernel~\cite{KZ14}. This two-sided kernel reduces to the sine kernel in the simplest case. We have also seen that the global density can be found explicitly and that it is expressed in terms of the so-called Fuss--Catalan distribution in a simple manner. Our result relies on an explicit double contour integral formulation of the correlation kernel (Proposition \ref{prop:kernel-finite}).
It is worth stressing that we could make full use of this double contour integral formulation and give an analytical proof of the global density.  In fact, following almost exactly the same steps introduced in  \cite{LWZ14}, it can be proven that the sine-kernel arises in the bulk 
and that Airy-kernel arises at the soft edge; cf. the proof of Theorems 1.1 and Theorem 1.3 as well as Remark 2 in~\cite{LWZ14}. The full details are beyond the scope of this paper, so let us only mention that a basic starting  point is to approximate the integrand by elementary functions and rewrite the kernel,  say the even part, as
\begin{multline}
\frac{1}{n \tilde\rho(x_0)}\Big( \sqrt{\frac{2}{n}} \Big)^{2M+1}K_{2n}^\text{even}
\bigg(
\Big( \sqrt{\frac{2}{n}} \Big)^{2M+1} \Big(\frac{x_0}{\sqrt{2}}+\frac{x}{ \tilde\rho(x_0)n}\Big), \Big( \sqrt{\frac{2}{n}} \Big)^{2M+1}(\frac{x_0}{\sqrt{2}}+\frac{x}{ \tilde\rho(x_0)n}\Big)\bigg)
\\ \sim \frac{\sqrt{2}}{|x_0|\tilde\rho(x_0)}\frac{1}{(2\pi i)^2}\int dt\int ds\,\frac{e^{n(g(s)-g(t))}}{s-t}  \Big|1+\frac{\sqrt{2}x}{ x_0\tilde\rho(x_0)n}\Big|^{2ns} \Big|1+\frac{\sqrt{2}y}{ x_0\tilde\rho(x_0)n}\Big|^{-2nt-1}
 \frac{h_{n}(s)}{h_{n}(t)},
\end{multline}
where the phase function is given by
\begin{equation}
g(z)=(2M+1)z-2(M+1)z\log z+z\log(z-1)-\log(1-z)+z\log x_{0}^2.
\end{equation}
Hence, the saddle point equation $g'(z)=0$ is exactly  expressed through the equation  \eqref{alg-eq1}.   We stress that the above  parametrisation representation  plays a key role in the proof of the sine-kernel  via the steepest decent method.

Finally, we emphasize that our construction of a Hermitised random product ensemble is based on a matrix transformation which maps  the space of polynomial ensembles onto itself (Theorem~\ref{T1}). This type of matrix transformation are important since they preserve exact solvability. Our  proof of Theorem~\ref{T1} is applicable  to the Hermitised product ensemble multiplied by a Gaussian matrix, crucially with the help of the   hyperbolic HCIZ integral over the pseudo-unitary group. However, it would be interesting to see whether this could be extended to the product ensemble multiplied by   other types of random matrices, say, truncated unitary matrices. For this, a possible way is to first extend the matrix integral formula stated  in \cite[Theorem 2.3]{KKS15} from  the unitary group to  the  pseudo-unitary case, and then perform  the same steps as in Section \ref{sec:fyodorov}. This will be an interesting and  challenging problem for us.

\paragraph{Acknowledgements}
We thank S. Kumar  for  useful discussions on his paper and M. Kieburg for comments on a first draft.
We acknowledge support by the Australian Research Council through grant DP170102028 (PJF),
the ARC Centre of Excellence for Mathematical and Statistical Frontiers (PJF,JRI), and  partially by ERC Advanced Grant \#338804,  the National Natural Science Foundation of China  \#11771417,  the Youth Innovation Promotion Association CAS  \#2017491, the Fundamental Research Funds for the Central Universities \#WK0010450002,   Anhui Provincial Natural Science Foundation \#1708085QA03 (DZL).   D.-Z. Liu is particularly grateful to L\'{a}szl\'{o} Erd\H{o}s  for funding   his one-year stay at  IST Austria.

\appendix

\section{Three different proofs of Theorem 1}

This appendix contains three separate proofs of Theorem~\ref{T1}. Each proof has its own merits and provides a different perspective on the matrix transformation.

The first proof is based on the method of additive rank-one deformations.
A benefit of this method is that it avoids group integrals of HCIZ type, and therefore might be a more suitable starting point for generalisations to studies of real or quaternion matrices.
The second proof is based on a theorem by Forrester and Rains~\cite{FR05} that gives the eigenvalue density of a Hermitised matrix product by means of an inverse double-sided Laplace transform. This idea is closely related to the spherical transforms used in the context of other matrix transformations~\cite{KK16,KR16}.
The third and final proof uses a generalisation of the HCIZ integral previously studied by Fyodorov~\cite{Fy02,FS02}.

The three proofs will be given in the three subsections below. However, before we start it is worth noting the following reduction result related to Theorem~\ref{T1}.

\begin{remark}\label{R1}
Suppose $n = N$, and consider the limit $a_{n_0+1} \to 0^+$. Recalling the ordering (\ref{as2}), inspection of (\ref{as1}) shows that its
leading contribution comes from a Laplace expansion via the top left entry of the second determinant.
This entry is in turn significant only for $x_{n_0+1} \to 0^+$, telling us that to leading order (\ref{as1}) with $n=N$ and
$a_{n_0+1} \to 0^+$ is equal to
\begin{equation}\label{as3}
\frac{e^{-x_{n_0+1}/a_{n_0+1}}}{|a_{n_0+1}|}\prod_{l=1}^N\frac1{(N-l)!}
\!\!\prod_{\substack{l=1\\l\neq n_0+1}}^N\!\! |a_l|^{-2}|x_l|
\!\!\!\prod_{\substack{1\leq j<k\leq N\\j,k\neq n_0+1}}\frac{x_k-x_j}{a_k-a_j}
\det\big[e^{-x_i/a_j}\big]_{i,j=1}^{n_0}\det\big[e^{-x_{i+n_0}/a_{j+n_0}}\big]_{i,j=2}^{N-n_0}.
\end{equation}
After relabelling, (\ref{as3}) is equivalent to (\ref{as1}) with $n=N-1$ times a Dirac delta function
corresponding to an eigenvalue at zero. Repeating this limiting procedure a total of $N - n$ times shows that
(\ref{as1}) in the case $n=N$ reduces to the general $n$ case.
\end{remark}


\subsection{First proof: Recursive structure using additive rank-one deformations}

In this section, we prove Theorem~\ref{T1} by induction. The induction step will be constructed using the method of rank-one deformations.
For the reader's convenience, we start by providing an outline of the proof; details will follow in the steps i), ii), and iii) below.

We first need some additional notation. Let $G^{(p)}$ denote the $p \times N$ matrix consisting of the first $p$ rows of the $n \times N$ complex Gaussian matrix $G$ and let $A^{(p)} = \diag(a_1,\dots, a_p)$. Define
\begin{equation}\label{Xp-def}
X^{(p)}=(G^{(p)})^\dagger A^{(p)} G^{(p)},\qquad p=1,\ldots,n.
\end{equation}
We see that $X^{(p)}$ is an $N\times N$ matrix and that $X^{(n)}=X$.
Moreover, $X^{(p)}$ has rank (less than or equal to) $p$ and we therefore know that it has (at most) $p$ non-zero eigenvalues; we will denote these eigenvalues $\lambda_k^{(p)}$ ($k=1,\ldots,p$).
The crucial observation is that $X^{(p)}$ is an additive rank-one deformation of $X^{(p-1)}$ for $p>1$. More precisely, we have
\begin{equation}\label{1rank-deform}
X^{(p)}=X^{(p-1)}+a_p\,\vec{x}\, \vec{x}^{\,\dagger},
\end{equation}
where $ \vec{x}$ is an $N \times 1$ column vector with standard complex Gaussian entries.
We will see below that if the eigenvalues of $X^{(p-1)}$ are known, then the rank-one deformation~\eqref{1rank-deform} can be used to find the conditional PDF for the eigenvalues of $X^{(p)}$. Let us denote this conditional PDF by
\begin{equation}\label{condition-prob}
Q^{n_0}_{p-1}(\{a_j\}_{j=1}^p;\{\lambda_j^{(p)}\}\,\vert\,\{\lambda_j^{(p-1)}\}),\qquad p=2,3,\ldots,n.
\end{equation}
It is clear that if the PDF $P^{n_0}_1(a_1;\lambda_1^{(1)})$ is known, then $P^{n_0}_p(\{a_j\};\{\lambda_j^{(p)}\})$ can be constructed recursively using
\begin{equation}\label{recursion}
P^{n_0}_{p}(\{a_j\};\{\lambda_j^{(p)}\})=\int_D \prod_{k}d\lambda_k^{(p-1)}\,
P^{n_0}_{p-1}(\{a_j\};\{\lambda_j^{(p-1)}\})\,
Q^{n_0}_{p-1}(\{a_j\};\{\lambda_j^{(p)}\}\,\vert\,\{\lambda_j^{(p-1)}\})
\end{equation}
for $p=1,\ldots,n$ and a suitable integration domain $D$. Thus, our proof can be divided into three steps:
\begin{itemize}
 \item[i)]
 Use the additive rank-one deformation~\eqref{1rank-deform} to find the conditional PDF~\eqref{condition-prob}.
 \item[ii)] Use the conditional PDF~\eqref{condition-prob} together with the recursion~\eqref{recursion} to show that if $P^{n_0}_{p-1}(\{a_j\};\{\lambda_j^{(p-1)}\})$ is given by~\eqref{as1} with $n=p-1$, then $P^{n_0}_{p}(\{a_j\};\{\lambda_j^{(p)}\})$ is given by~\eqref{as1} with $n=p$.
 \item[iii)] Show that $P^{n_0}_{p=1}(a_1;\lambda_1^{(1)})$ is given by~\eqref{as1} with $n=1$.
\end{itemize}
We will look at these three steps separately below.

\paragraph{i)} We want to  find the eigenvalues of $X^{(p)}$ using the rank-one deformation~\eqref{1rank-deform} assuming that the eigenvalues of $X^{(p-1)}$ are known. The matrix $X^{(p-1)}$ has (at most) $p-1$ non-zero eigenvalues. We assume that these eigenvalue are pairwise distinct and ordered as
\begin{equation}\label{9}
-\infty<\lambda_1^{(p-1)}<\cdots<\lambda_{n_0}^{(p-1)}<0< \lambda_{n_0+1}^{(p-1)}<\cdots<\lambda_{p-1}^{(p-1)}<\infty,
\end{equation}
i.e. $X^{(p-1)}$ has $p-n_0-1$ positive eigenvalues if $p>n_0$ and no positive eigenvalues if $p\leq n_0$.

If $p<N$ then the matrix $X^{(p)}$ must have an eigenvalue equal to zero with multiplicity $N-p$ (or higher).
The remaining $p$ eigenvalues will be random variables which are non-zero and have multiplicity one (almost surely),
since the deformation~\eqref{1rank-deform} includes a Gaussian vector $x$.
Thus, we know from~\cite{FR05} that the eigenvalues of $X^{(p)}$ are given as solutions to the secular equation
\begin{equation}\label{6a}
0=1-a_p\bigg(\frac{q_0}{\lambda}+\sum_{j=1}^{p-1}\frac{q_j}{\lambda-\lambda_j^{(p-1)}}\bigg),
\end{equation}
where, with  $\Gamma[\alpha,\beta]$ denoting a gamma-distibuted variable with shape parameter $\alpha$ and rate parameter $\beta$, each $q_j$ is a random variable given by
\begin{equation}
q_0\stackrel{d}{=}\Gamma[N-p+1,1]\qquad\text{or}\qquad
q_j\stackrel{d}{=}\Gamma[1,1]\qquad\text{for}\quad j=1,\dots,p-1.
\end{equation}
Furthermore, the eigenvalues of $X^{(p)}$ must be interlaced with the eigenvalues $X^{(p-1)}$, i.e. interlaced with $\{0\}\cup\{\lambda_k^{(p-1)}\}_{k=1}^{p-1}$. This interlacing may be verified by sketching the plot of the secular equation~\eqref{6a} as a function of $\lambda$. Moreover, we note that whether the interlacing starts from the left or from the right depends on whether $a_p$ is negative or positive, or equivalently on whether $p\leq n_0$ or $p>n_0$, cf.~\eqref{as1}. For $p\leq n_0$ we have the interlacing
\begin{align}\label{domain1}
-\infty<\lambda_1^{(p)}<\lambda_1^{(p-1)}<\lambda_{2}^{(p)}
<\cdots<\lambda_{p-1}^{(p)}<\lambda_{p-1}^{(p-1)}<\lambda_{p}^{(p)}<0,
\end{align}
while for $p>n_0$ we have the interlacing
\begin{align}\label{domain2}
-\infty<\lambda_1^{(p-1)}<\lambda_1^{(p)}<\cdots<\lambda_{n_0}^{(p)}<0
<\lambda_{n_0+1}^{(p)}<\cdots<\lambda_{p-1}^{(p-1)}<\lambda_{p}^{(p)}<\infty.
\end{align}
Subject to these interlacings, we read off from~\cite[Cor.~3]{FR05} that the corresponding
conditional PDF~\eqref{condition-prob} is given by
\begin{multline}\label{7a}
Q^{n_0}_{p-1}(\{a_j\};\{\lambda_j^{(p)}\}\,\vert\,\{\lambda_k^{(p-1)}\})=\\
\frac{1}{|a_p|^N(N-p)!}
\frac{\prod_{i}(\lambda_i^{(p)})^{N-p}\,e^{-\lambda_i^{(p)}/a_p}}
{\prod_{k}(\lambda_k^{(p-1)})^{N-p+1}\,e^{-\lambda_k^{(p-1)}/a_p}}
\frac{\prod_{i<j}(\lambda_j^{(p)}-\lambda_i^{(p)})}
{\prod_{k<\ell}(\lambda_\ell^{(p-1)}-\lambda_k^{(p-1)})}
\end{multline}
with indices $1\leq i,j\leq p$ and $1\leq k,\ell\leq p-1$.

\paragraph{ii)}
We can now turn to the recursive formula~\eqref{recursion}.
The probability density~\eqref{as1} with $n=p-1$ is given by
\begin{multline}
P^{n_0}_{p-1}(\{a_j\};\{\lambda_j^{(p-1)}\})=
\prod_{l=1}^{p-1}\frac1{|a_l|}\frac{(\lambda^{(p-1)}_l/a_l)^{N-p+1}}{(N-l)!}
\prod_{1\leq j<k\leq p-1}\frac{\lambda^{(p-1)}_k-\lambda^{(p-1)}_j}{a_k-a_j}\\
\times\det\big[e^{-\lambda^{(p-1)}_i/a_j}\big]_{i,j=1}^{n_0}
\det\big[e^{-\lambda^{(p-1)}_{i}/a_{j}}\big]_{i,j=n_0+1}^{p-1}.
\end{multline}
Furthermore, we know from step~i) that the conditional PDF is given by~\eqref{7a}
and that the integration domain is given by either~\eqref{domain1} or~\eqref{domain2} depending on whether $p\leq n_0$ or $p>n_0$. Considering the case $p\leq n_0$, we have the recursion
\begin{multline}\label{proof-recursion}
\int \prod_{k}d\lambda_k^{(p-1)}\,
P^{n_0}_{p-1}(\{a_j\};\{\lambda_j^{(p-1)}\})\,
Q^{n_0}_{p-1}(\{a_j\};\{\lambda_j^{(p)}\}\,\vert\,\{\lambda_j^{(p-1)}\})=\\
a_p^{-p+1}\prod_{l=1}^{p}\frac1{|a_l|}\frac{(\lambda_l^{(p)}/a_l)^{N-p}e^{-\lambda_i^{(p)}/a_p}}{(N-l)!}\prod_{l=1}^{p-1}\frac1{a_l}
\prod_{1\leq j<k\leq p-1}\frac{1}{a_k-a_j}
\prod_{i<j}(\lambda_j^{(p)}-\lambda_i^{(p)})\\
\times\int \prod_{k}d\lambda_k^{(p-1)}\,
\det\big[e^{-\lambda^{(p-1)}_i(1/a_j-1/a_p)}\big]_{i,j=1}^{p-1}.
\end{multline}
We note that there is only one determinant since $p<n_0$.
Let us focus on the integral on the last line in~\eqref{proof-recursion}. We see that
\begin{align}\label{proof-det}
\int \prod_{k}d\lambda_k^{(p-1)}\,
\det\big[e^{-\lambda^{(p-1)}_i(a_j^{-1}-a_p^{-1})}\big]_{i,j=1}^{p-1}
=&\det\bigg[\int_{\lambda_{i}^{(p)}}^{\lambda_{i+1}^{(p)}}e^{-x(a_{j}^{-1}-a_p^{-1})}\,dx\bigg]_{i,j=1}^{p-1} \nonumber \\
=&\det\bigg[\int_{\lambda^{(p)}_1}^{\lambda_{i+1}^{(p)}}e^{-x(a_{j}^{-1}-a_p^{-1})}\,dx\bigg]_{i,j=1}^{p-1}
\end{align}
with integration domain on the right-hand side on the first line given by~\eqref{domain1}.
The first equality in~\eqref{proof-det} follows by shifting the integration inside the determinant,
while the second equality follows by a standard row manipulation.
Performing the integral within the determinant on the last line~\eqref{proof-det}, we see that
\begin{align}
\det\bigg[\int_{\lambda^{(p)}_1}^{\lambda_{i+1}^{(p)}}e^{-x(a_{j}^{-1}-a_p^{-1})}\,dx\bigg]_{i,j=1}^{p-1}
&=a_p^{p-1}\prod_{j=1}^{p-1}\frac{a_j}{a_j-a_p}
\det\big[e^{-\lambda_{i+1}^{(p)}(a_{j}^{-1}-a_p^{-1})}-e^{-\lambda_{1}^{(p)}(a_{j}^{-1}-a_p^{-1})}\big]_{i,j=1}^{p-1} \nonumber \\
&=a_p^{p-1}\prod_{j=1}^{p-1}\frac{a_j}{a_j-a_p}\prod_{\ell=1}^p e^{\lambda_\ell^{(p)}/a_p}
\det\big[e^{-\lambda_{i}^{(p)}/a_{j}}\big]_{i,j=1}^{p},
\end{align}
where the last inequality can be understood by applying elementary row operations of adding multiples of the first row to the rows below so as to get zero entries in the final column, expect for the first entry, then Laplace expanding by that entry.

Finally using this evaluation of the integral from~\eqref{proof-recursion} verifies the recursion for $p\leq n_0$. The verification for $p>n_0$ follows the same lines.

\paragraph{iii)}
It only remains to show that $P^{n_0}_{p=1}(a_1;\lambda_1^{(1)})$ is given by~\eqref{as1} with $n=1$. There are two cases $n_0=0$ and $n_0=1$, but they are both immediate since we are considering scalars.

\subsection{Second proof: Limit of inverse Laplace transform expression}

This second proof of our main theorem starts by looking at the eigenvalues of a more general matrix
\begin{equation}\label{GAG+B}
G^\dagger AG+B,
\end{equation}
where $G$ be an $n\times N$ complex Gaussian random matrix, while $A$ and $B$ are Hermitian matrices.
Due to unitary invariance we can in fact choose $A$ and $B$ to be diagonal,
say $A=\diag(a_1,\ldots,a_n)$ and $B=\diag(b_1,\ldots,b_N)$,
without loss of generality.

It is evident that this eigenvalue problem reduces to that of Theorem~\ref{T1}, when
\begin{equation}\label{s1b}
b_1,\dots, b_N \to 0,
\end{equation}
This observation is crucial, since it was shown by Forrester and Rains~\cite[Thm. 6]{FR05} that
(assuming the eigenvalues of $A$ and $B$ are pairwise distinct)
the eigenvalue PDF for the matrix~\eqref{GAG+B} can be written as
\begin{equation}\label{s1}
\ePDF(G^\dagger AG+B)=
\frac{1}{N!}\frac{\det[x_i^{j-1}]_{i,j=1}^N}{\det[b_i^{j-1}]_{i,j=1}^N}
\det \Big [ \mathcal L^{-1}[\det(\mathbb I + A s)^{-1}](x_i - b_j)
\Big ]_{i,j=1}^N.
\end{equation}
where $\ePDF(M)$ is a short-hand notation for the eigenvalue density for a random matrix $M$ and $\mathcal L^{-1}$ denotes  the inverse two-sided Laplace transform,
\begin{equation}\label{s2}
\mathcal L^{-1}[f(s)](x):=\lim_{\tau\to0^+}\int_{-i\infty}^{+i\infty}\frac{ds}{2\pi i}\,e^{sx+\tau s^2/2}f(s).
\end{equation}
In some sense, this result is more general than the one we are trying to prove, but it is also far less explicit
and therefore less useful for our purposes.
Thus, the strategy to prove Theorem~\ref{T1} presented in this subsection is
to show that given~\eqref{as} then~\eqref{s1} reduces to \eqref{as1} in the limit~\eqref{s1b}. Here~\eqref{as1} refers to the PDF
with the $N$ non-zero eigenvalues. However, we know from Remark~\ref{R1} that the $n=N$ case of~\eqref{as1}
implies the general $n$ case. It is therefore sufficient for us to set $n=N$ in (\ref{s1}) and to show
that in the limit (\ref{s1b}) the case $n=N$ of (\ref{as1}) appears.

For this purpose, we begin by noting that with $f(s) = \det(\mathbb I + A s)^{-1}$, the limit $\tau \to 0^+$
in (\ref{s2}) can be taken inside the integral, telling us that
\begin{equation}
\mathcal L^{-1}[\det(\mathbb I+As)^{-1}](x_i-b_j)
=\int_{-i\infty}^{+i\infty}\frac{ds}{2\pi i}\,\frac{e^{s(x_i-b_j)}}{\prod_{k=1}^N(1+a_ks)}.
\end{equation}
Our evaluation of this contour integral follows the standard procedure of first closing the contour and then using the residue theorem. Due to the inequalities~\eqref{as2}, we see that the contour must be closed in the positive half-plane for $i=1,\ldots,n_0$, which according to the inequalities~\eqref{as1} picks up contributions from $n_0$ simple poles located at $-1/a_1,\ldots,-1/a_{n_0}$. For $i=n_0+1,\ldots,N$ the contour must be closed in the negative half-plane picking up contributions from the remaining $N-n_0$ poles located at $-1/a_{n_0+1},\ldots,-1/a_{N}$.
After substituting this straightforward evaluation of the contour integral into the last determinant in (\ref{s1}), the limit~\eqref{s1b} may be found by successive use of L'H\^opital's rule. This yields
\begin{equation}\label{s1ea}
\ePDF(G^\dagger AG)=
\prod_{k=0}^N\frac1{k!}\det [ x_i^{j-1} ]_{i,j=1}^N
\!\!\!\!\!\!\!\!\sum_{\substack{1\leq k_1,\dots,k_{n_0}\leq n_0\\n_0+1\leq k_{n_0+1},\dots,k_N\leq N}}\!\!\!\!\!\!\!\!
\det\bigg[\frac1{|a_{k_i}|}\frac{a_{k_i}^{N-j}e^{-x_i/a_{k_i}}}{\prod_{l=1,\,l\neq k}^N(a_{k_i}-a_l)}\bigg]_{i,j=1}^N.
\end{equation}
The latter determinant in this expression may be simplified considerably by noting that the only factor inside determinant which depends on both index $i$ (through $k_i$) and index $j$ is $a_{k_i}^{N-j}$, while the only factor depending on both index $i$ and index $k_i$ is $e^{-x_i/a_{k_i}}$. Thus, upon expansion and reordering of products we see that
\begin{equation}
\det\bigg[\frac1{|a_{k_i}|}\frac{a_{k_i}^{N-j}e^{-x_i/a_{k_i}}}{\prod_{l=1,\,l\neq k_i}^N(a_{k_i}-a_l)}\bigg]_{i,j=1}^N
=\prod_{i=1}^N\frac{e^{-x_i/a_{k_i}}}{|a_i|}\prod_{1\leq i<j\leq N}\frac1{(a_j-a_i)^2}\det\big[a_{k_i}^{j-1}\big]_{i,j=1}^N.
\end{equation}
Thus, the eigenvalue PDF becomes
\begin{multline}\label{long-expression-2}
\ePDF(G^\dagger AG)=\prod_{k=0}^N\frac1{k!}\det [ x_i^{j-1} ]_{i,j=1}^N
\prod_{i=1}^N\frac{1}{|a_i|}\prod_{1\leq i<j\leq N}\frac1{(a_j-a_i)^2}\\
\times\sum_{\substack{1\leq k_1,\dots,k_{n_0}\leq n_0\\n_0+1\leq k_{n_0+1},\dots,k_N\leq N}}
\prod_{i=1}^Ne^{-x_i/a_{k_i}}\,\det\big[a_{k_i}^{j-1}\big]_{i,j=1}^N.
\end{multline}
In order to evaluate the sums on the second line in~\eqref{long-expression-2}, we note that the indices satisfy $k_1,\dots,k_{n_0}<k_{n_0+1},\dots,k_N$ and that the determinant is anti-symmetric in both $\{k_1,\ldots,k_{n_0}\}$ and $\{k_{n_0+1},\ldots,k_N\}$. This allow us to write
\begin{equation}\label{triple-vandermonde}
\sum_{\substack{1\leq k_1,\dots,k_{n_0}\leq n_0\\n_0+1\leq k_{n_0+1},\dots,k_N\leq N}}\!\!\!\!
\prod_{i=1}^Ne^{-x_i/a_{k_i}}\,\det\big[a_{k_i}^{j-1}\big]_{i,j=1}^N=
\det\big[a_{i}^{j-1}\big]_{i,j=1}^N\det[e^{-x_i/a_j}]_{i,j=1}^{n_0}\det[e^{-x_i/a_j}]_{i,j=n_0+1}^{N}.
\end{equation}
The first determinant on the right-hand side in~\eqref{triple-vandermonde} and the first determinant in~\eqref{long-expression-2} are both Vandermonde determinants, so the eigenvalue PDF~\eqref{long-expression-2} becomes
\begin{equation}
\ePDF(G^\dagger AG)=\prod_{k=0}^N\frac1{k!}\prod_{i=1}^N\frac{1}{|a_i|}
\prod_{1\leq i<j\leq N}\frac{x_j-x_i}{a_j-a_i}
\det[e^{-x_i/a_j}]_{i,j=1}^{n_0}\det[e^{-x_i/a_j}]_{i,j=n_0+1}^{N},
\end{equation}
which we recognise as the desired statement~\eqref{as} with $n=N$.

\subsection{Third proof: Matrix integral over the pseudo-unitary group}
\label{sec:fyodorov}

For this third proof of Theorem~\ref{T1} we will again restrict our attention to the case $n=N$.
The Gaussian matrix $G$ specified in Theorem \ref{T1} has distribution
\begin{equation}\label{Gp}
\Big(\frac1{\pi}\Big)^{N^2}e^{-\tr G^\dagger G}(dG),
\end{equation}
where $(dG)$ is the Lebesgue measure on the space of complex $N\times N$ matrices.
It is a standard result from random matrix theory that the positive semi-definite Hermitian matrix $\tilde{W}= G G^\dagger$ is distributed according to
\begin{equation}\label{Wt}
\prod_{k=0}^{N-1}\frac1{\pi^kk!}\,e^{-\tr\tilde{W}}(d\tilde W),
\end{equation}
where $(d\tilde W)$ is the Lebesgue measure on the space of Hermitian matrices subject to the constraint that $\tilde W$ is positive semi-definite. This may be seen by decomposing the matrix $G$ using a polar decomposition, i.e. $G=U\tilde W^{1/2}$ with $U$ unitary. Making this change of variables and integrating over the unitary degrees of freedom contribute an extra factor, $2^{-N} {\rm vol} \, U(N)$, to the normalisation;
see e.g.~\cite{Fo10}.

The proof presented in this section is based on an integration formula for the pseudo-unitary group with a pseudo-metric tensor $\eta$ determined according to the number of positive (negative) eigenvalues of the matrix $A$.
More precisely, with $A$ as specified in Theorem \ref{T1}, we define
\begin{equation}\label{S}
A_+=\diag(|a_1|,\ldots,|a_N|) \qquad\text{and}\qquad
\eta=\eta^N_{n_0}=\diag(\,\underbrace{-1,\ldots,-1}_{n_0},\underbrace{+1,\ldots,+1}_{N-n_0}\,),
\end{equation}
such that $A=A_+\eta=\eta A_+$. In fact, we have the more general relation $A=A_+^p\eta A_+^q$ with $p+q=1$, since the matrix $A$ is assumed to be non-singular.

Next, introduce the matrix $Z= A_+^{1/2}\tilde{W}A_+^{1/2}$ where $\tilde W$ is an $N\times N$ matrix distributed according to~\eqref{Wt}.
Since both $Z$ and $\tilde{W}$ are complex Hermitian random matrices, we have $(dZ)=(\det A_+)^N(d \tilde{W})$,
see e.g.~\cite[Eq.~(1.35)]{Fo10}.
Thus, using that $A_+$ is invertible we know from~\eqref{Wt} that the distribution of $Z$ is equal to
 \begin{equation}\label{Z1}
\prod_{k=0}^{N-1}\frac1{\pi^kk!}\,\frac{e^{-\tr A_+^{-1} Z}}{(\det A_+)^{N}}(d Z).
 \end{equation}
In terms of $Z$ we may define $\tilde{Z}=\eta Z=\eta A_+^{1/2}GG^\dagger A_+^{1/2}$. This matrix is important, since its eigenvalues are identical to those of
\begin{equation}\label{Z=GAG}
G^\dagger A_+^{1/2}\eta A_+^{1/2}G=G^\dagger AG.
\end{equation}
Here, the right-hand side is the matrix of interest for Theorem~\ref{T1} and the equality is a simple consequence of the definition~\eqref{S}.
Moreover, since $(d \tilde{Z}) = (d Z)$ (the action of $\eta$ on $Z$ is only to change the sign of some of
 the entries of $Z$), we read off from (\ref{Z1}) that the distribution of $\tilde{Z}$ is equal to
 \begin{equation}\label{Z2}
\prod_{k=0}^{N-1}\frac1{\pi^kk!}\,\frac{e^{-\tr A^{-1}\tilde{Z}}}{(\det A_+)^{N}}(d\tilde{Z}),
 \end{equation}
where it is further required that $\eta\tilde{Z}$ is positive semi-definite.

It is a known result that if $Z$ is a positive definite matrix, then the matrix $\eta Z$ has exactly $n_0$ negative eigenvalues and $N-n_0$ positive eigenvalues~\cite{PS82}.
Furthermore, the matrix $\tilde{Z}$ can be diagonalised using a pseudo-unitary similarity transformation, i.e. there exists a matrix $V\in U(\eta)$ such that
 \begin{equation}\label{Z2a}
\tilde{Z} = V L V^{-1},
 \end{equation}
where $L=\diag(x_1,\dots, x_N)$ is a real diagonal matrix. We recall that the pseudo-unitary group is defined as
\[
U(\eta)=\{V\in GL(N,\mathbb C)\,\vert\,V^\dagger \eta V=V\eta V^\dagger=\eta\}.
\]
The Jacobian associated to the change of variables~\eqref{Z2a} is~\cite{PS82}
\begin{equation}\label{Z2b}
(d\tilde{Z})=(V^{-1} d V)\prod_{1\leq i<j\leq N}(x_j-x_i)^2\prod_{l=1}^N dx_l,
\end{equation}
where $(V^{-1}dV)$ is the Haar measure on the pseudo-unitary group.
For the measures (\ref{Z2b}) to be in one-to-one correspondence, it is necessary to restrict the overall phase of each eigenvector, or equivalently require that $V \in U(\eta)/U(1)^N$.

Substituting~\eqref{Z2b} into (\ref{Z2}) shows that the eigenvalue PDF of $\tilde{Z}$
(or equivalently of $G^\dagger A G$) is equal to
 \begin{equation}\label{Z3}
 P_{n_0, n}( \{ a_j \}_{j=1}^N; \{ x_j \}_{j=1}^N ) =
 \prod_{k=0}^{N-1}\frac1{\pi^kk!} \prod_{1\leq i<j\leq N}(x_j-x_i)^2
 \int e^{-\tr A^{-1} V L V^{-1} }  \,\frac{(V^{-1} d V)}{\det(A_+)^{N}}.
 \end{equation}
We note that for $n_0=0$ and $n_0=N$, the signature $\eta$ becomes proportional to the identity and the group integral in~\eqref{Z3} reduces to the well-known HCIZ integral~\eqref{HCIZ}.
The generalisation of the HCIZ integral from an integral over the unitary group to an integral over the pseudo-unitary group (i.e. $0<n_0<N$) has been studied by Fyodorov~\cite{Fy02,FS02}, who showed that
 \begin{equation}\label{Z4}
 \int_{U(\eta)/U(1)^N}e^{-\tr AVBV^{-1}}\,(V^{-1}dV)= K_{N,n_0}
 \frac{\det[e^{-a_ib_j}]_{i,j=1}^{n_0}\det[e^{-a_{i+n_0}b_{j+n_0}}]_{i,j=1}^{N-n_0}}
 {\prod_{1\leq i<j\leq N} (a_j-a_i)(b_j-b_i)}
 \end{equation}
with $K_{N,n_0}$ denoting an undetermined proportionality constant, and $A=\diag(a_1,\ldots,a_N)$ and $B=\diag(b_1,\ldots,b_N)$ denoting diagonal matrices subject to the constraints
\begin{align}
a_1<\cdots<a_{n_0}&<0<a_{n_0+1}<\cdots<a_N, \nonumber \\
b_1<\cdots<b_{n_0}&<0<b_{n_0+1}<\cdots<b_N. \label{constraints}
\end{align}
The constraints~\eqref{constraints} must be included to ensure convergence of the group integral on the right-hand side in~\eqref{Z4} for an $0<n_0<N$. This is necessary since the pseudo-unitary group is non-compact except for $n_0=0$ or $n_0=N$ in which case the aforementioned constraints may be ignored.

Now, using the integration formula~\eqref{Z4} to evaluate the group integral in~\eqref{Z3}, we see that
 \begin{multline}
 P_{n_0, n}( \{ a_j \}_{j=1}^N; \{ x_j \}_{j=1}^N ) =
 K_{N,n_0}\prod_{k=0}^{N-1}\frac1{\pi^kk!\,|a_{k+1}|}
 \prod_{1\leq i<j\leq N}\frac{x_j-x_i}{a_j-a_i}\\
 \times \det[e^{-x_i/a_j}]_{i,j=1 }^{n_0}\det[e^{-x_{i+n_0}/a_{j+n_0}}]_{i,j=1}^{N-n_0},
 \end{multline}
which agrees with the $n=N$ case of~\eqref{as1} provided that $K_{N,n_0}=\pi^{N(N-1)/2}$.
We note that the proportionality constant is independent of $n_0$. Moreover, the cases $n_0=0$ and $n_0=N$ (where the group integral is over $U(N)/U(1)^N$) are consistent with the known proportionality constant from the HCIZ integral; recall that our choice of measure is not normalised to unity rather we have ${\rm vol}\,U(N)/U(1)^N=\pi^{N(N-1)/2}/\prod_{j=1}^N\Gamma(j)$.

\end{document}